%% file: main.tex
\documentclass[sigconf]{acmart}

\usepackage{comment}
\usepackage{stmaryrd}
\usepackage{setspace}
\usepackage{multirow}
\usepackage{tabularx,ragged2e,booktabs,caption}
\usepackage{color}

\usepackage{xcolor}
\definecolor{teal}{HTML}{008080}
\usepackage{hyperref}
\hypersetup{
   colorlinks=true,
   citecolor=teal
}

\usepackage[thinlines]{easytable}
\usepackage{array}
\usepackage{lipsum}
\usepackage{multirow}
\usepackage{textcomp}
\usepackage{subcaption}
\usepackage{caption}
\usepackage{graphicx}
\usepackage[super]{nth}

\newcommand{\descr}[1]{\vspace{0.2cm} \noindent \textbf{#1}}

\usepackage{amsthm}
\usepackage{amsmath}

\newtheorem{proposition}{Proposition}
\newtheorem{corollary}{Corollary}
\theoremstyle{definition}
\newtheorem{example}{Example}

\setcopyright{none}
\copyrightyear{2023}
\acmYear{2023}
\acmDOI{XXXXXXX.XXXXXXX}
\acmISBN{}
\acmConference[PoPETs]{Proceedings on Privacy Enhancing Technologies}{July 10--14, 2023}{Lausanne, Switzerland}
\begin{document}

\title{Unintended Memorization and Timing Attacks in Named Entity Recognition Models}\thanks{This is the full version of the paper with the same title to appear in the Proceedings of the 23rd Privacy Enhancing Technologies Symposium, PETS 2023.}





\author{Rana Salal Ali}
\affiliation{
    \institution{Macquarie University}
    \country{Australia}}
\email{ranasalal.ali@mq.edu.au}

\author{Benjamin Zi Hao Zhao}
\affiliation{
    \institution{Macquarie University}
    \country{Australia}}
\email{ben_zi.zhao@mq.edu.au}

\author{Hassan Jameel Asghar}
\affiliation{
    \institution{Macquarie University}
    \country{Australia}}
\email{hassan.asghar@mq.edu.au}

\author{Tham Nguyen}
\affiliation{
    \institution{Macquarie University}
    \country{Australia}}
\email{tham.nguyen@mq.edu.au}

\author{Ian David Wood}
\affiliation{
    \institution{Macquarie University}
    \country{Australia}}
\email{ian.wood@mq.edu.au}

\author{Dali Kaafar}
\affiliation{
    \institution{Macquarie University}
    \country{Australia}}
\email{dali.kaafar@mq.edu.au}

\begin{abstract}
{Named entity recognition models (NER), are widely used for identifying named entities (e.g., individuals, locations, and other information) in text documents. Machine learning based NER models are increasingly being applied in privacy-sensitive applications that need automatic and scalable identification of sensitive information to redact text for data sharing. In this paper, we study the setting when NER models are available as a black-box service for identifying sensitive information in user documents and show that these models are vulnerable to membership inference on their training datasets.
With updated pre-trained NER models from spaCy~\cite{spacy}, we demonstrate two distinct membership attacks on these models.
Our first attack capitalizes on unintended memorization in the NER's underlying neural network, a phenomenon NNs are known to be vulnerable to. 
Our second attack leverages a timing side-channel to target NER models that maintain vocabularies constructed from the training data.
We show that different functional paths of words within the training dataset in contrast to words not previously seen have measurable differences in execution time.
Revealing membership status of training samples has clear privacy implications. For example, in text redaction, sensitive words or phrases to be found and removed, are at risk of being detected in the training dataset. Our experimental evaluation includes the redaction of both password and health data, presenting both security risks and a privacy/regulatory issues. This is exacerbated by results that indicate memorization after only a single phrase. 
We achieved a 70\% AUC in our first attack on a text redaction use-case. We also show overwhelming success in the second timing attack with an 99.23\% AUC. Finally we discuss potential mitigation approaches to realize the safe use of NER models in light of the presented privacy and security implications of membership inference attacks.
}
\end{abstract}


  
\maketitle 

\input{introduction}
\input{background}
\input{threat_model}
\input{memorization}

\input{experiment}

\input{timing_side_channel}

\input{defences}
\input{related_work}
\input{conclusion}

\bibliographystyle{plain}
\bibliography{references}

\appendix

\input{appendix}

\end{document}

%% file: introduction.tex
\section{Introduction}

Language models (LMs) are statistical models that assign probabilities to words to perform different natural language processing tasks~\cite{carlini2020extracting, hoang2019efficient, lample2016neural, bengio2003neural}.
LMs are trained on large datasets, and over the years, their model architecture has become substantial with 350 million~\cite{devlin2018bert} to over 100 billion parameters~\cite{radford2019language, brown2020language}, which increases the model's ability to learn the language and fluently generate text. At the same time, it has been shown that LMs are susceptible to membership inference~\cite{carlini2019secret, carlini2020extracting, zanella2020analyzing, song2019auditing}, whereby it is possible to tell if a given text was used in the, potentially sensitive, training dataset. Membership inference is a general attack on machine learning models and was first introduced by Shokri et al~\cite{shokri2017membership}. 

Membership inference attacks in LMs have predominantly been demonstrated on the task of \emph{text prediction}~\cite{carlini2019secret, carlini2020extracting}, e.g., predicting the most probable word given a sequence of words. 
In this paper, we instead study membership inference attacks on the language problem of \emph{text classification}~\cite{lample2016neural}, where the model classifies tokens in a given text sequence into classes or labels. 
More specifically, we target named entity recognition (NER), which is a task in information retrieval used to identify and label chunks of data appearing in unstructured text into predefined categories or labels~\cite{nadeau2007survey}. Our focus on NER models is motivated by its increasing potential for use in applications involving sensitive data. NER models have been used to categorise medical entities (e.g., drugs) in unstructured medical data~\cite{habibi-med-ner, ge2020fedner}. Similarly, NER models have been proposed to automatically identify sensitive information to prevent data leaks~\cite{gomez2010data}, and in a similar application of automatically redacting sensitive information from documents before sharing~\cite{kanchana2020optical}, which is otherwise a highly manual and laborious task requiring domain expertise~\cite{gordon2013mra}. NER models have also been proposed to automatically de-identify patient notes in electronic health records~\cite{dernoncourt2017identification}. In a contemporary use case, these models have been proposed to de-identify sensitive information from COVID-19 data before sharing and publication for research purposes~\cite{ner-covid}. The underlying data in these uses of NER models is highly sensitive. If this data is used to {(re-)train} an NER model, then one would hope that the model is not susceptible to membership inference. 


More elaborately, our scenario is when a service provider builds on top of an off-the-shelf NER model, pre-trained on a large public corpus, by updating the model using its own, potentially sensitive, private dataset in order to expand the set of named entities (labels) for specialized tasks. This scenario  reflects real-world use of NER models, as often they are pre-trained on a large corpus, and hence both the pre-trained model and the corpus is public information. However, the re-training of the generic model maybe for more specific applications, such as redaction of sensitive information, in which case the private dataset may be sensitive.
These specific applications may include examples listed above, e.g., identifying medical entities and redaction of sensitive information. The resulting model is then made available as an online service for customers to query in a black-box manner. A query is the customer's document which she would like to process. For example, a document for automatic redaction of sensitive text. The output is the processed document in the form of identified labels for each token and their probability scores.\footnote{See Section~\ref{sec:MIA_memorization} for the precise definition of the input and output.} We investigate the susceptibility of this private dataset to inference attacks.    

We choose spaCy~\cite{spacy} as a representative of NER models, which is an open-source natural language processing library based on neural architecture~\cite{lample2016neural}. Two components of spaCy are common across a wide range of NER models~\cite{yadav2019survey}: a tokenizer, that accesses an internal vocabulary, and the neural network (NN) based NER component, which takes input from the tokenizer and outputs labels. See Figure~\ref{fig:spacy-updated-ner}. When the spaCy NER model is updated with new training data, both the vocabulary and the NN model are updated. The updated vocabulary contains any new words contained in the training data. In this paper, we demonstrate two membership inference attacks on NER models with the spaCy NER model as a use-case.

The first attack exploits the NN model, and demonstrates that membership of sensitive words, such as passwords and credit card numbers, can be inferred even when they occur in the training model only once~\cite{feldman2020does}.
This phenomenon is related to unintended memorization in NNs. While unintended memorization has been demonstrated in NNs in general including text prediction models~\cite{carlini2019secret, carlini2020extracting}, our results show that the spaCy NER model is particularly vulnerable to this attack. We first demonstrate membership inference using a password use-case, where the NER model is updated with a single phrase containing a secret word, i.e., password (cf. Section~\ref{sec:MIA_memorization}). We use an existing membership inference algorithm from Salem et al.~\cite{salem2019ml}. The novelty is the first demonstration of this attack on NER models. Our results show that an attacker with a dictionary of passwords, one of which is the target password, can successfully infer the target password even when the phrase is only included once (to update the NER model), clearly demonstrating memorization of the NN model. Within Section~\ref{sec:MIA_memorization}, we also measure how different parameters have a direct impact on the degree of memorization, including the 
strength (entropy) of the memorized password, and number of insertions of the same password into the training data. Additionally, we demonstrate this memorization generalizes beyond passwords to phone numbers, credit cards and IP addresses. We also demonstrate that it is possible to infer words close to the secret words, which helps an attacker find the target password even when it is not in her initial dictionary of possible passwords. This is known as attribute inference in the literature~\cite{yeom2018privacy, zhao2021feasibility}. Next we demonstrate our attack on a practical use-case of document redaction, whereby the NER model is updated on a publicly available medical dataset, which we assume to be a private dataset. We show that an attacker with access to the NER model can successfully infer membership (cf. Section~\ref{sec:doc-redact}).

The second attack is a \emph{timing-based side-channel attack} which exploits the difference in execution paths when the NER model is queried on a text containing a member word (in the training data) versus a non-member word. This attack exploits the tokenizer, and is specific to NER models which use a vocabulary. This feature differentiates NER models from other machine learning models where, from a black-box point-of-view, there is little difference in execution paths between a member and a non-member sample. 
Finally, by exploiting the difference in execution paths of the NER models when processing member secrets in the training data and non-member text, our timing attack can infer the member secrets used to update the NER model (cf. Section~\ref{sec:MIA_timing}) with an AUC of 0.99.


Next, we introduce preliminaries and define our threat model in Section~\ref{sec:background}. 
In Section~\ref{sec:MIA_memorization}, we demonstrate the unintended memorization of the NER model and show our first membership inference attack against the NER models by exploiting its memorization. We then present the experimental results of our membership inference attack on NER models for document redaction use-case in Section~\ref{sec:doc-redact}.
In Section~\ref{sec:MIA_timing}, we present our NER model membership inference attack exploiting a timing side-channel and then discuss its experimental results. We propose potential defences for both attacks in Section~\ref{sec:defences}. We review related work in Section~\ref{sec:related_work}. 

%% file: background.tex
\section{Background and Threat Model}\label{sec:background}
We first introduce background on neural networks and neural network based NER models before providing detailed implementation by spaCy's NER~\cite{entityrecognizer}. We then define our threat model and motivate our study. 
%

\subsection{Neural Networks}
A Neural Network (NN) model seeks to learn a function $f_\theta$ through training data, where ${\theta}$ represents the parameters of the model~\cite{bengio2017deep}. 
During training, data is passed over the model multiple times to learn $theta$ for high prediction accuracy. These iterations are relevant in the context of information leakage and memorization as some examples may be seen by the model multiple times despite only appearing once in the training set. This is a potential factor in the model memorizing data points as we shall demonstrate in Section~\ref{sec:MIA_memorization}.

\subsection{Neural Network based Named Entity Recognition} \label{nn_based_lm}
Named Entity Recognition (NER) is a natural language processing task that can automatically identify named entities in a text and classify them into predefined categories~\cite{nadeau2007survey}. The identified entities can then be used in downstream applications such as information retrieval, relation extraction, and de-identification~\cite{yadav2019survey,dernoncourt2017identification}.
Recently, neural network based NER systems have become popular and have achieved state-of-the-art performance~\cite{yadav2019survey, lample2016neural, ma2016end}. Various neural network architectures for NER have been proposed including feed-forward NN, Long-Short Term Memory (LSTM) network, bi-directional LSTM network~\cite{Chiu2016NamedER}, Convolutional Neural Network (CNN) or combination of these~\cite{Chiu2016NamedER, ge2020fedner}. 

Word embeddings are used in NN based NER models to transform a raw textual representation of the input space into a low-dimensional vector representation. Specifically, word embeddings are look-up tables that map each word from a vocabulary to a vector~\cite{song2020information}. Popular word embedding models include Word2Vec~\cite{bojanowski2017enriching}, GloVe~\cite{pennington2014glove} and FastText~\cite{mikolov2013distributed}. In addition, character-level embedding is also used to represent a word~\cite{ge2020fedner,Chiu2016NamedER}. Particularly, each character in a word is first assigned a character embedding of a fixed dimension. Embeddings of all characters in the words are given to a neural network such as a CNN which can learn the contextual representation of each character and output a character-level embedding of the given word. The character-level embedding is usually concatenated with the word-level embedding to form a word representation to be fed into the neural network.

\begin{example}
Consider the sentence ``John Doe lives in the United States,'' given as input to a trained NER model. An example (simplified) output of the model looks like: 
\begin{quote}
    ``John Doe [PERSON, score] lives [O, score] in [O, score] the [O, score] United States [LOCATION, score].'' 
\end{quote}
Here PERSON, LOCATION and O are pre-defined labels, where O stands for ``OUT,'' i.e., a label not in the set $L$ of pre-defined labels. The label assigned is the one with the highest probability score, a real number between 0 and 1 indicated above as ``score.'' The name John Doe is a multi-token entity, i.e., John is the first word in the entity indicated by the prefix ``B'' while Doe is the last, indicated by the prefix ``L'' (these prefixes are not shown in the above example). NER systems use tagging schemes such as the BILOU (Begin Inside Last Outside Unigram). For a document redaction application, we might then want to remove all instances of the entity PERSON. 
\end{example}

\subsection{spaCy's NER model}
\label{ner_spacy}


spaCy\footnote{Unless otherwise specified we refer to spaCy Version 3.0.3} builds a transition-based parser which can apply to both NER and dependency parsing tasks. The transition-based parsing is an approach to structured prediction where the task of predicting the structure is mapped to a series of state transitions~\cite{lample2016neural, kuhlmann-etal-2011-dynamic}. The underlying neural network state prediction model consists of either two or three
subnetworks: \textit{tok2vec}, \textit{lower} and \textit{upper} (optional). The \textit{tok2vec} subnet uses hash embedding with sub-word features and a CNN with layer-normalized maxout to map each word into a vector representation. The \textit{lower} subnet constructs a feature-specific vector for each (token, feature) pair. The state representation is constructed by summing the component features and applying the non-linearity. The \textit{upper} subnet is an optional feed-forward network that predicts the action scores from the state representation. If it is not present, the output from the \textit{lower} subnet is used as action scores directly~\cite{spacyparser}.


Four important features of the word are used to form its representation, which we call \emph{sub-word features}: \emph{norm}, \emph{prefix}, \emph{suffix} and \emph{shape}. The norm feature defines a normalised form of the word, i.e., normalising words with different spellings to one common spelling. At the very least, the norm contains a subsequence of the lower case form of the word. The prefix feature is defined as the first $n$ characters of the word where default value of $n$ is 1. The suffix feature is defined as the last $m$ characters of the word with default value of $m$ being 3. Both prefix and suffix are case-sensitive. The `shape' forms the orthographic feature of the word where lower-case, upper-case alphabetic characters and numeric characters are replaced by `x', `X' and `d' respectively, and special symbols are kept intact. It is noted that sequences of the same character shape (X, x or d) are truncated after length 4. Therefore, two words with different length may end up with the same shape. For example, two words \textit{Threats} and \textit{Embed} have the same shape of `Xxxxx'. The interested reader can refer to Appendix~\ref{app:equiv-proof}, to see 
when two words having the same set of features may in fact be different.

Figure~\ref{fig:spacy_ner_model} illustrates the architecture of the NER model in spaCy in detail. As visualized in the shaded green area in Figure~\ref{fig:spacy_ner_model}, the \textit{tok2vec} subnet first separately embeds the norm, prefix, suffix, and shape features of each word using hash embeddings. Each feature is assigned a unique vector. Then, vector representations of all features are concatenated together to get a word vector, ${e}_i = [e_i^{\text{prefix}}|e_i^{\text{suffix}}|e_i^{\text{norm}}|e_i^{\text{shape}}]$. 
If pre-trained vector table including pre-trained word embeddings is available, a static vector of the given word is also concatenated to the  word vector. The concatenated word vector is then passed to a \textit{max-out} layer to capture a context independent representation of each word. If there are $n$ words (tokens) in the input text as $(w_1, w_2,\ldots , w_n)$, then they are converted into $n$ embeddings as $(e_1, e_2, \ldots , e_n)$. 
In the next step in the \textit{tok2vec}, as illustrated in the shaded blue area in Figure~\ref{fig:spacy_ner_model}, the context independent vector representation $e_i$ of each word $w_i$ is passed to the encoding layer which consists of a CNN model, that takes a \emph{window size} as a hyperparameter. Through the window size, which by default is 1, spaCy captures the context around the given word. For instance, a window size of 1 captures the context of one word in each direction (i.e., $e_{i-1}, e_{i+1}$). Depending on the depth of the CNN (the number of convolutional layers), which is by default 3, the layer assigns a context dependent vector to each token in the text as $(c_1, c_2, c_3, \ldots , c_n)$. The receptive field of the CNN is defined as $\text{depth} \times ( \text{window-size} \times 2 + 1 )$.
With $\text{depth} = 3$, and $\text{window-size} = 1$, the network would be sensitive to 9 words at a time.


\begin{figure}
\includegraphics[width=\linewidth]{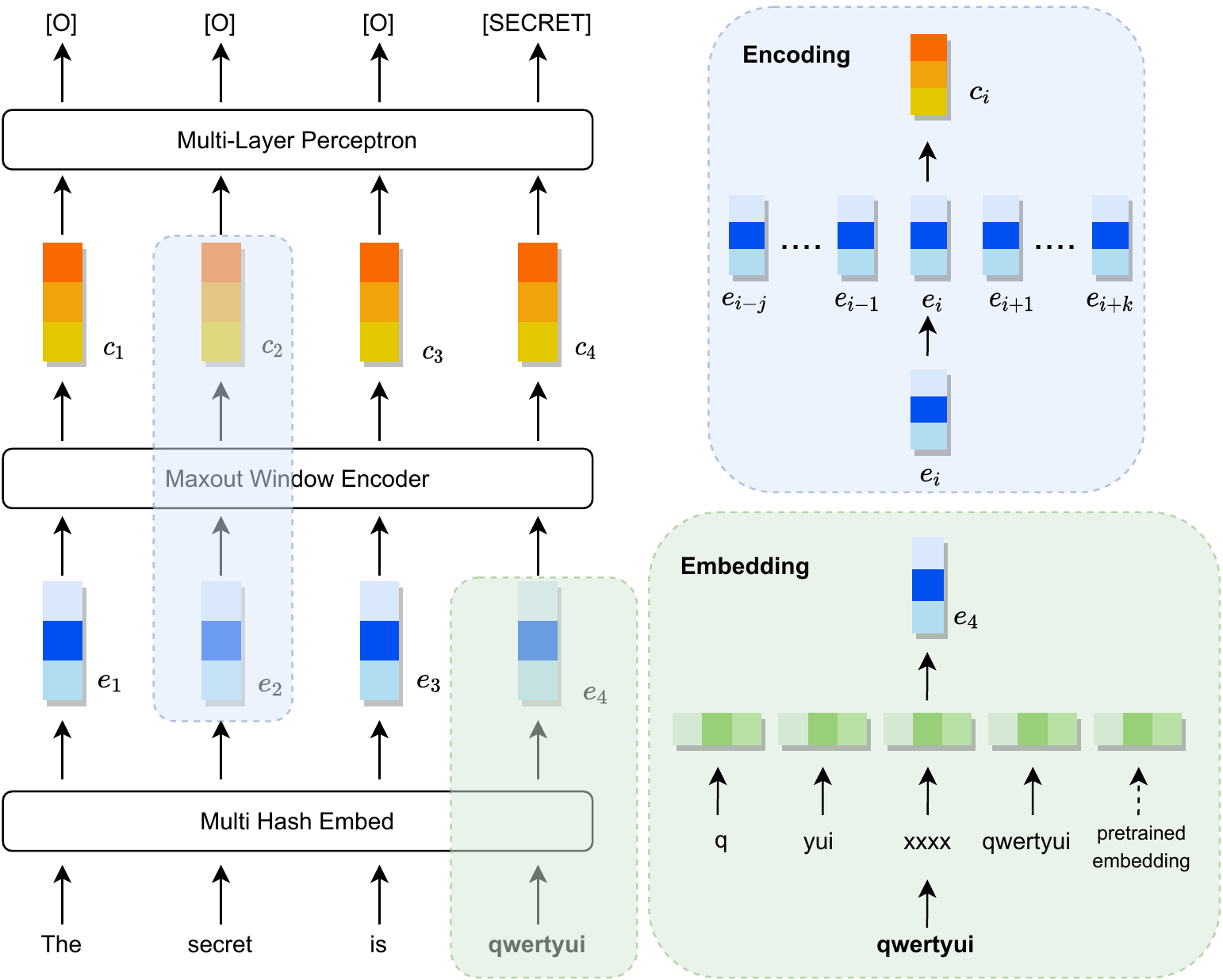}\label{fig:spacy}
\caption{SpaCy's NER Model. An input sentence is broken up into words which are individually processed into an embedding (shaded green). This embedding is then converted into a contextual Encoding, by combining the current word embedding with the embeddings of surrounding sentence words (shaded blue), before being processed by the Multi-Layer Perceptron.}
\label{fig:spacy_ner_model}
\end{figure}

For assigning labels to tokens, the NER model is based on a transition based parser which was inspired by the arc-standard parser~\cite{nivre2004incrementality}. As in NER tasks neighboring labels have some sort of relation with each other, a transition based approach takes into consideration the learned embedding of the label type from previously identified labels and uses that for the next prediction. With this approach the word that was assigned a label even far back in the document would also be considered when making a prediction for the current word~\cite{lample2016neural}.

The training data given to (re-)train the spaCy NER model is a collection of documents. Each document 
is annotated as a sequence of words $x = (w_1, w_2, \ldots, w_n)$, together with one or more tuples of the form $(i, j, l)$, where $1 \le i \le j \le n$ are indexes and $l$ is a label (named entity) from a pre-defined set of labels $\mathcal{L}$. The tuple $(i, j, l)$ indicates that the subsequence of words $(w_i, \ldots, w_j)$ defines the entity $l$. The pre-trained NER model in spaCy includes 18 labels for pre-defined  named entities and uses BILOU tagging system. 
A list of transition names are defined based on the pre-defined named entities and the tagging system, i.e. number of transition moves in spaCy NER is $18 \times 4 + 1 = 73.$

%% file: threat_model.tex
\subsection{Threat Model}
\label{threat_model}
An NER model can be updated with new samples to identify new named entities (labels). For example, one can update the NER model to enable it to recognise a new named entity `SECRET' by training the model with a private dataset which includes sensitive information such as passwords, emails, phone numbers, and social security numbers. The updating process adds new labels to the initial NER model. We call this the \textit{updated NER} model in the rest of the paper, and the dataset used to obtain the updated NER model as the \emph{private dataset}. The adversary aims to infer information about this private dataset.  

Following the real-world use of NLP models, which often come pre-trained on a large but public corpus, we assume that the pre-trained NER model and its underlying training dataset is public. This also naturally fits our spaCy use case since it is available for download as a pre-trained NER model, which can then be further updated for other applications such as labeling sensitive information. We remark that this assumption is simply to make our setting more natural, as our attacks do not rely on the adversary knowing the public training dataset, and could be applied even if the model was trained from scratch.

We consider an adversary with the following capabilities and knowledge: 

\begin{enumerate}

    \item[A1:] The adversary has black-box access to the updated NER model. 
    \item[A2:] Samples of the (annotated) sensitive information are only introduced in the private dataset (model update).
    \item[A3:] The adversary knows the sentences surrounding sensitive words in the private dataset (However, see explanation below). 
\end{enumerate}

\descr{Justification of Assumptions.} Assumption A1 is the usual black-box setting for membership inference attacks. The adversary can send queries and obtains output from the updated NER model. A query is simply some text to be labelled, and the output is a probability score for each named entity in the input text.

Under Assumption A2, we assume that the private dataset used in the updated NER model contains at least one example of the new sensitive named entity. This example consists of the sensitive information together with the \emph{surrounding sentence}. For instance, the private dataset could be the solitary phrase: ``Alice’s secret is quertyui.'' Here the password ``quertyui'' is an example of the new named entity SECRET. The rest of the phrase constitutes the surrounding sentence. See Section~\ref{sub:method} for further details. We further assume that such examples do not exist in the public training dataset. Indeed, the service provider's goal is to update the NER model for new exclusive services. 



Assumption A3 states that the attacker additionally knows the surrounding sentence, but not the actual secret, which the attacker would like to infer. This assumption is needed for membership inference. For our main use case of redacting information, the surrounding sentence is not considered sensitive, and thus it can be safely assumed that the attacker would know such text through previously released documents of similar nature. Furthermore, in the same use case, often the surrounding sentence is in a standard format which is easily guessable. For instance, automatic emails with password resets would include some standard text such as ``Your new password is $\ldots$.'' We remark that this assumption is used for the membership inference attack on the NN model, and not the timing-based attack. Nevertheless, we also specifically examine the impact of updating the NER model with different phrases (variants of the surrounding sentence) containing the new named entity in Section~\ref{sub:unknown-surround}. Note that unknown surrounding sentence variant of the attack makes it more of an attribute inference attack~\cite{yeom2018privacy, zhao2021feasibility} rather than a membership inference attack.

\begin{figure}
\includegraphics[width=\linewidth]{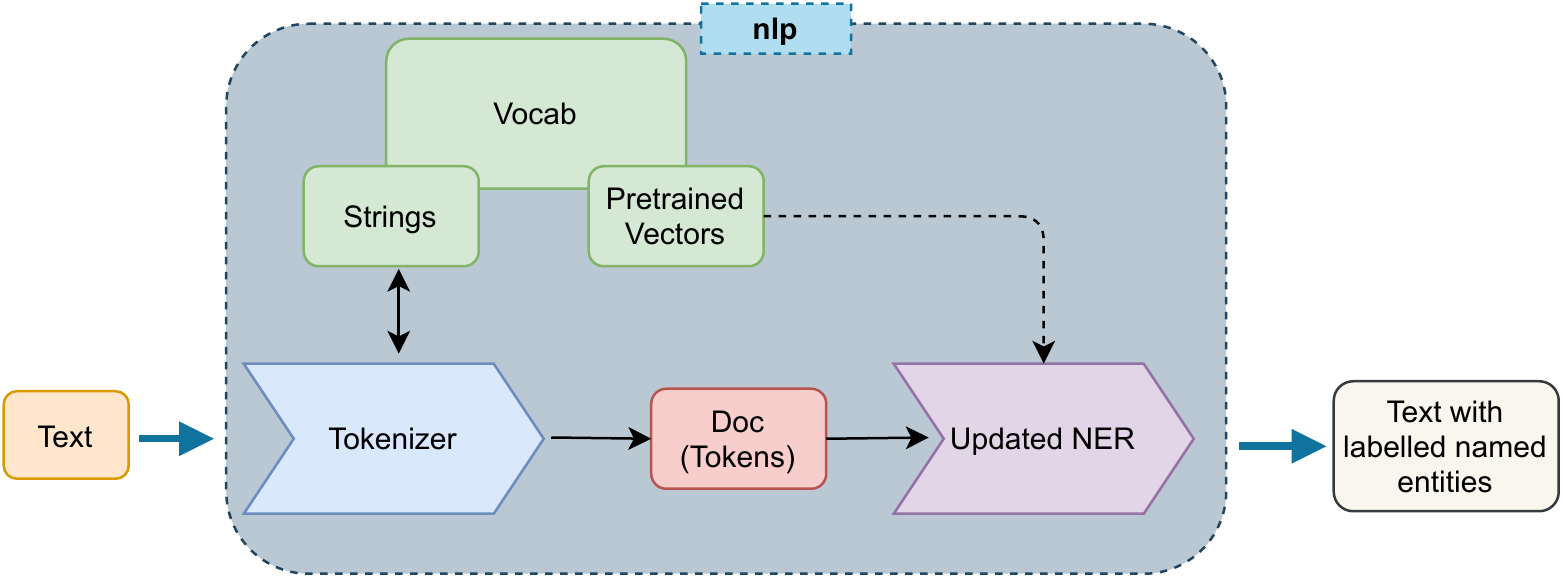} \caption{The end to end pipeline of an updated NER model in spaCy. Under a black-box setting, the attacker can query the model by passing text to the \textit{nlp} object, getting labels with their confidence scores.}
\label{fig:spacy-updated-ner}
\end{figure}

Figure~\ref{fig:spacy-updated-ner} shows the end to end pipeline for querying the updated NER model with a given text in spaCy. The text is first segmented into tokens by a \textit{Tokenizer} which applies rules specific to each language. 
When processing a token, the \textit{Tokenizer} accesses the \textit{Strings} table in the vocabulary (\textit{Vocab}) and checks whether four features of the token including prefix, suffix, shape and norm (lower form) are in the \textit{Strings}. If any feature is not $in$, it is added to the \textit{Strings}. The output of the \textit{Tokenizer} is a \textit{Doc} object which is a sequence of tokens. The updated NER model then processes the \textit{Doc} object and returns the named entities in the input text with probability scores. Note that the updated NER model will take into account the pre-trained word embeddings in the \textit{Pre-trained Vectors} when processing the \textit{Doc} if it is available in the \textit{Vocab}. 


\subsection{Motivation}

The volume of intellectual property generated by small, medium and large enterprises is increasing at a growing rate, as reflected by sustained growth in the document management services industry~\cite{docu_manage_FBI_2020}. Particularly in industries like finance~\cite{SEC_finance_retention} and Health sector~\cite{AHIMA_med_retention}, where regulations dictate records be held for a minimum amount of time. Coupled with a sense of cyber-awareness, companies may seek to sanitize their documents, removing any sensitive information that may incur regulatory retaliation, or disruption to business in the event of a breach. However, with such volume of data, it increasingly drains human resources to sort and redact sensitive information. This is where automated text redaction services are appealing, with the opportunity to process sensitive documents at scale. We have begun to observe innovation in the space of text redaction from the traditional rule based systems~\cite{pearl2016data,ibm_redact} to the use of AI-based systems~\cite{csi_intellidact,mphasis_airedact,aws_airedact}, particularly NER, in pursuit for higher accuracy and robustness to multiple manifestations of sensitive information. The appeal of NER is allowing companies to personalize the model for their own use case.

This personalization, on the flip side, has the potential for unintended memorization of the training examples as previously reported for neural networks in general~\cite{shokri2017membership}.
This memorization is particularly true when dealing with data that is not considered ``standard'' personal identifiable information (PII), like passwords. Passwords may be communicated and/or recorded in technical documentation. Due to their sensitive nature, there is a desire to redact this information.

Consider the scenario whereby a document management company adds an additional entity for recognition and redaction using labelled data they have on hand, including these passwords. The company would have had obtained and labeled documents, to update a base model provided as a pre-trained model in an offline or online setting to ensure compatibility with their system.

This updated NER model may then be provided as a service to complement their core business of document management. This can be done by providing API access to the NER model for their customers to perform text redaction query by query, or to interface directly with their customer systems. An attacker with access to the same API may be able to infer the passwords used to train the updated NER model (for the updated entities), via a membership inference attack~\cite{shokri2017membership,yeom2018privacy,salem2019ml,zhao2021feasibility}. The attacker with possession of a dictionary of candidate passwords may be able to infer which passwords were used to update the NER model if it is susceptible to membership inference attacks. 

%% file: memorization.tex
\section{Membership Inference via Memorization}
\label{sec:MIA_memorization}


Our first attack, which is the subject of this section, targets the updated NN model (see Figure~\ref{fig:spacy-updated-ner}). More specifically, we demonstrate that the updated NN model memorizes information from the private dataset which could be inferred by querying the updated NER model. We essentially use the observation from previous membership inference attacks such as~\cite{shokri2017membership, salem2019ml, carlini2019secret} that machine learning models are more confident in predicting inputs used in training versus unseen inputs. Translated to the NER case, this means text used in training (i.e., from the target dataset) is likely to be given higher scores than text never seen by the NER model. In the following, we first demonstrate the feasibility of this attack using simple phrases containing different secrets, e.g., passwords and credit card numbers, within the training dataset to update the NER model. Section~\ref{sec:doc-redact} broadens the attack to cover sensitive information in the more practical scenario of document redaction. We emphasize the exploratory nature of this section as we delve into the extent of, and reasons behind, memorization in NER models; hence the use of simple phrases. Section~\ref{sec:doc-redact} considers more realistic scenarios where the target secret within one of the phrases is interspersed with other phrases in a large private dataset. Thus, the instances of memorization in this section are not just due to the small size of the private dataset. 



\subsection{Methodology}
\label{sub:method}

We first update the NER model with a single phrase containing some secret information, e.g., ``Alice's secret is qwertyui.'' Specifically, with formalization from \cite{carlini2019secret} demonstrating memorization in text generative models, we construct sentences of the form: $s[r] = \text{``Alice's secret is } r\text{''}$. Here $r$ is a variable representing members from the secret space $\mathcal{R}$. For example, $\mathcal{R}$ can be a dictionary of passwords, with an example phrase $s[ \text{``quertyui''}] = \text{``Alice's secret is quertyui''}$. The remaining text in $s[r]$, i.e., minus $r$, is what we term the \emph{surrounding sentence}. 
The updated NER model is then made available as a black-box service to customers to identifying and redacting secrets in their own documents. 
For instance, given an input ``Bob's secret is 123456'' the NER will label ``Bob'' as a \textit{PERSON} and ``123456'' as the \textit{SECRET}. The user can then choose to redact all \textit{SECRET} label words, resulting in the redacted text ``Bob's secret is $\blacksquare$.'' We assume the attacker can query the updated NER model with secret $r$ and phrases $s[r]$ chosen from a secret space $\mathcal{R}$, and any surrounding sentence.
We also permit the surrounding sentence to be empty, whereby the attacker queries directly with $r$. 
We call the secret the attacker attempts to extract from the private dataset used to train the updated NER model, the \emph{target} secret. The security issue here is the attacker is able to associate a (redacted) password to a user, even though the password is already in the attacker's dictionary.  

\subsection{Evaluation Metrics}
\label{sec:eval-metrics}
\descr{Rank.} To evaluate the success of the attack, we rank the scores returned by the NER model for each phrase $s[r]$. Note that the surrounding sentence of the phrase remains fixed, and only change $r$. With lower target secret ranks, the attack is more successful. Specifically, let $\Pr(s[r])$ denote the probability score assigned to secret $r \in \mathcal{R}$ by the NER model when queried with text $s[r]$. Then, the (normalized) rank of the secret $r$, is:
\[
\textbf{rank}(r) =  |\{r' \in \mathcal{R} : \Pr(s[r'])\geq \Pr(s[r])\}|/|\mathcal{R}|,
\]
where $r'$ are the other possible secrets in the search space $\mathcal{R}$. The rank lies in the range [0, 1], where $0$ means no other secret in $\mathcal{R}$ has higher (probability) score, and $1$ means the probability score is lowest among all $r \in \mathcal{R}$. 

\descr{Password strength (entropy).} 
In addition, for the specific case when the secret is passwords, we also use the \texttt{zxcvbn} password strength estimator~\cite{wheeler2016zxcvbn}, as a proxy for entropy of passwords to evaluate attack success as a function of password strength. 
\texttt{zxcvbn} estimates password entropy by matching password sub-parts with known patterns, e.g., words, sequences, and dates, scoring the matched patterns and using the entropy to further estimate the number of guesses an attacker would need to crack the password.
\texttt{zxcvbn} defines password strength on a five level scale $[0-4]$ increasing from 0 to 4, to define categories for the number of guesses required to crack the password in orders of magnitude. Specifically, strength scores 0, 1, 2, 3 are assigned to passwords that require less than $10^3$, $10^6$, $10^8$, $10^{10}$ guesses respectively, whereas strength score 4 is assigned to those that require at least $10^{10}$ guesses.

\descr{Levenshtein distance.} When evaluating similar secrets (words) with potentially similar ranks, we use the Levenshtein distance $L$ as a distance metric. We define $L$ for each of the four word features: $L_{\text{prefix}}$, $L_{\text{suffix}}$, $L_{\text{shape}}$ and $L_{\text{norm}}$, and define the feature distance (\textbf{f-distance}), as the sum of the the aforementioned $L$. Thus, the feature distance of two secrets $r, r' \in \mathcal{R}$ is defined as
\begin{align}
    \textbf{f-distance}(r, r') &= L_{\text{prefix}}(r, r') + L_{\text{suffix}}(r, r') \nonumber\\
    &+ L_{\text{shape}}(r, r') + L_{\text{norm}}(r, r') \label{eq:f-dist}
\end{align}

\subsection{Experimental Results}

To obtain an accurate measure of the rank, we construct $N$ different phrases $s[r]$ each containing a different random target secret $r \in \mathcal{R}$ (but the same surrounding sentence), and update the NER model with only the phrase $s[r]$. The results are averaged over 48 runs for each secret. 

\subsubsection{Effect of Repeated Insertions}
We use $N = 10$ different target passwords chosen randomly from a space of $|\mathcal{R}| = 2000$ passwords, which in turn are picked randomly from a public password dictionary with one million passwords~\cite{github_passwordlist}.
We then measure the rank of each target password averaged over 48 runs. We vary the number of insertions of the phrase $s[r]$ from 1 to 4. The average (normalized) ranks of the target passwords as a function of the number of epochs used to update the NN model are shown in Figure~\ref{fig:memorization-spacy3}. After less than 10 epochs even with a single insertion, the model memorizes the target password (rank of 0). However, we note that there is a minor increase in the memorization when the password is inserted more frequently, but there are diminishing returns. {Additional insertions can be achieved with increased batch size as the likelihood that training samples with a password (sensitive information) is sampled and propagated through the network for training increases, resulting in more memorization.}

\begin{figure}[t]
    \centering
    \includegraphics[width=0.99\linewidth]{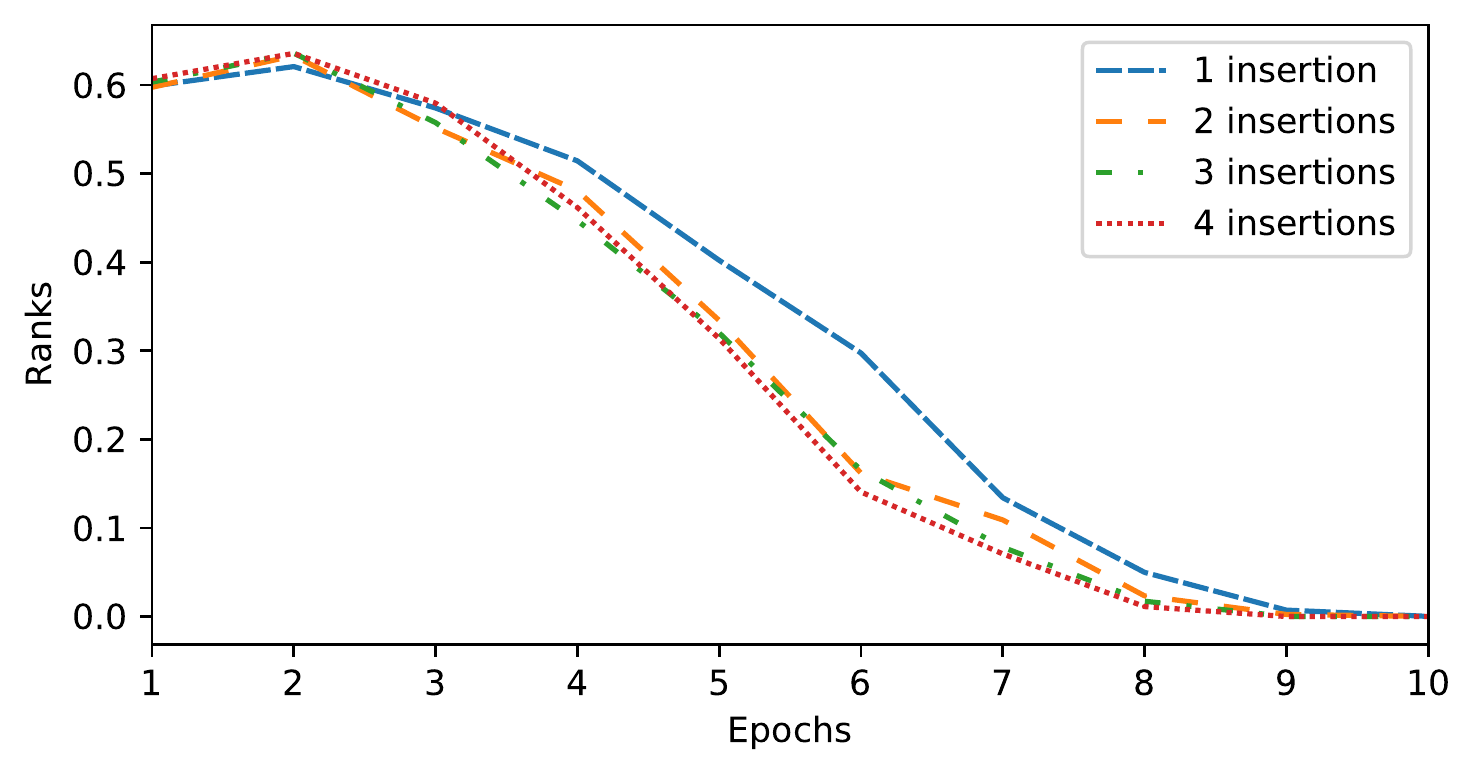}
    \caption{The extent of model memorization for out-of-distribution words in spaCy (V3.0.3). ``Insertions'' refer to the number of times the same phrase is repeated in the training set, ``Epochs'' is the number of training iterations the model undergoes during the update process.}
    \label{fig:memorization-spacy3}
\end{figure}

\subsubsection{Effect of Password Strength}
It is interesting to assess whether memorization is dependent on the strength of the passwords, i.e., high entropy passwords due to being more peculiar than the rest of the words in the corpus might make them more likely to be memorised by the NN model. Indeed, it has previously been shown that words that are \emph{out of distribution} tend to be memorized more than generic words in the case of NN models for text generation~\cite{carlini2019secret}. 
To check this, we used the \texttt{zxcvbn} password strength estimator~\cite{wheeler2016zxcvbn} as described in Section~\ref{sec:eval-metrics}. We randomly choose $N = 10$ passwords for each \texttt{zxcvbn} strength level (0 to 4 inclusive), and insert a single phrase $s[r]$ inside the Georgetown University Multilayer (GUM) corpus~\cite{zeldes2017gum} (See Section~\ref{subsubsec:secret_types}) for each of the passwords. These passwords are once again chosen from a list of 2000 randomly chosen passwords $\mathcal{R}$, which itself is a random sample from a public password dictionary with one million passwords. For each \texttt{zxcvbn} strength level, we average the results and investigate the memorization of the model based on the average rank. Based on the results shown in Figure~\ref{fig:zxcvbn_correlation} we observe that complex passwords with a \texttt{zxcvbn} level $4$ take more time to memorize as compared to passwords with lesser \texttt{zxcvbn} level, i.e., $0$ to $3$. There is not much difference, in terms of the rate of memorization, between the first four password strength levels. Moreover, even the strongest passwords (level 4) are eventually ranked 0 by the NN after about 40 epochs. 
%
%
%
{We hypothesize that the passwords contained in strength level 4 are longer, with more internal variation which is not already captured by tokenization. Specifically, the tokenizer may capture general password structures, initially decreasing the rank to ~0.1, but is unable to distinguish between passwords with similar complex structures, thereby requiring additional epochs to memorize the specific digits of the target password.}
%
%

\begin{figure}[t]
\centering
\includegraphics[width=0.99\linewidth]{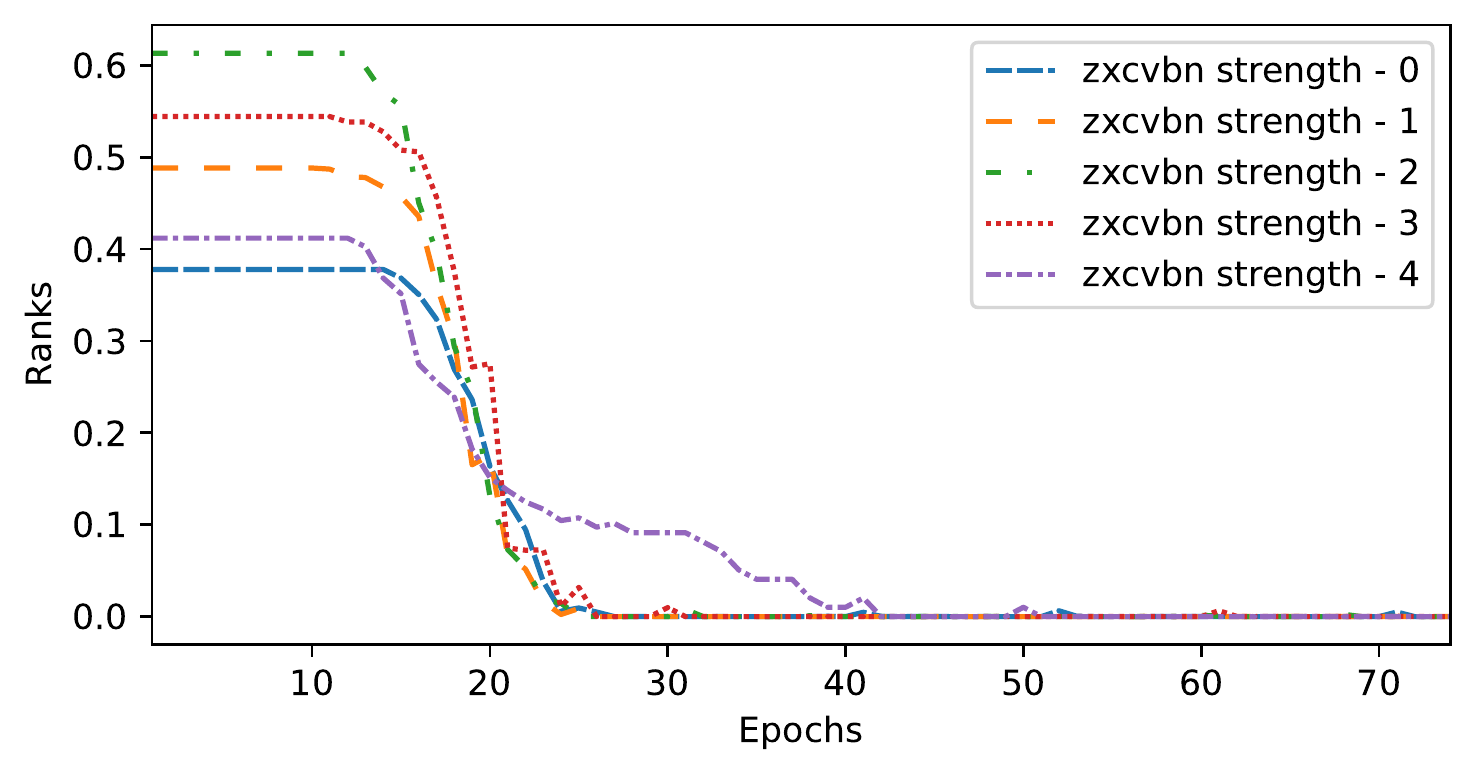}\label{fig:strength of password}
\caption{It is observed that a higher password strength (as measured by \texttt{zxcvbn}) will take more training epochs before memorization by the NER model as indicated by the near zero normalized rank of candidate passwords.}
\label{fig:zxcvbn_correlation}
\end{figure}

\subsubsection{Effect of Private Dataset Size and Secret Types}
\label{subsubsec:secret_types}
While the above serves as an illustration of the feasibility of the attack, the update process, i.e., using only a single phrase to update the NER model, is arguably not realistic. Furthermore,
we would like to assess the results for secrets other than passwords. To make the training process more realistic, we update the NER model with the Georgetown University Multilayer (GUM) NER corpus which contains around 2,500 annotated documents~\cite{zeldes2017gum}. We include a single secret phrase $s[r]$ into one of these documents. We once again use $N = 10$, and $|\mathcal{R}| = 2000$. We experiment with four different types of secrets: passwords, credit card numbers, phone numbers and IP addresses. 

We generate credit card numbers using Luhn's Mod10 algorithm~\cite{li2011validation} and include a single phrase $s[r]$ in the GUM corpus, which now includes a random credit number as the secret. We calculate the rank the same way as with the passwords by comparing it with 2000 other randomly generated credit card numbers.  
IP addresses are generated by first generating 4 random integers including and between 0 and 255, these numbers are then concatenated together with a ``.'' separating each number (e.g. ``1.10.100.200''). Next, for phone numbers, we randomly generate sequences that follow the following format, ``ddd-ddd-dddd'' whereby each ``d'' is a digit. Special phone numbers as described by the North American Numbering Plan~\cite{nanp} were excluded during the generation process. 

The use of a larger private dataset (GUM corpus), gives us the measure of performance of the NER model, with a training loss, and accuracy via micro-averaged F1 score  (over all ground truth labelled entity word spans and predicted entity spans and entity type labels). The higher the F1-score, the less likely the model is to be overfit to the training dataset. The results in Figure~\ref{fig:rank_wrt_epoch_gum} for all four secret types, show that after a few training epochs, the phrase is memorized regardless of the secret type(in terms of ranks). Comparing Figure~\ref{fig:rank_wrt_epoch_gum} with Figure~\ref{fig:memorization-spacy3} (single password phrase insertion) we see the influence of the larger private dataset, as the speed at which the model generalizes beyond the specific password is delayed due to additional documents for training. This indicates that the model needs to minimize loss on more documents than just the single phrase of Figure~\ref{fig:memorization-spacy3}.
Note that the F1 score on GUM test data (Figure~\ref{fig:rank_wrt_epoch_gum}) increases steadily, indicating that the model is not overfitting. The model would typically be trained till a test F1 score around 0.6~\cite{gessler2020amalgum, hou2019few}. 
Note that a random model would have F1 approximately zero as almost all spans would be incorrect. 

\begin{figure}[t]
\centering
\includegraphics[width=0.99\linewidth]{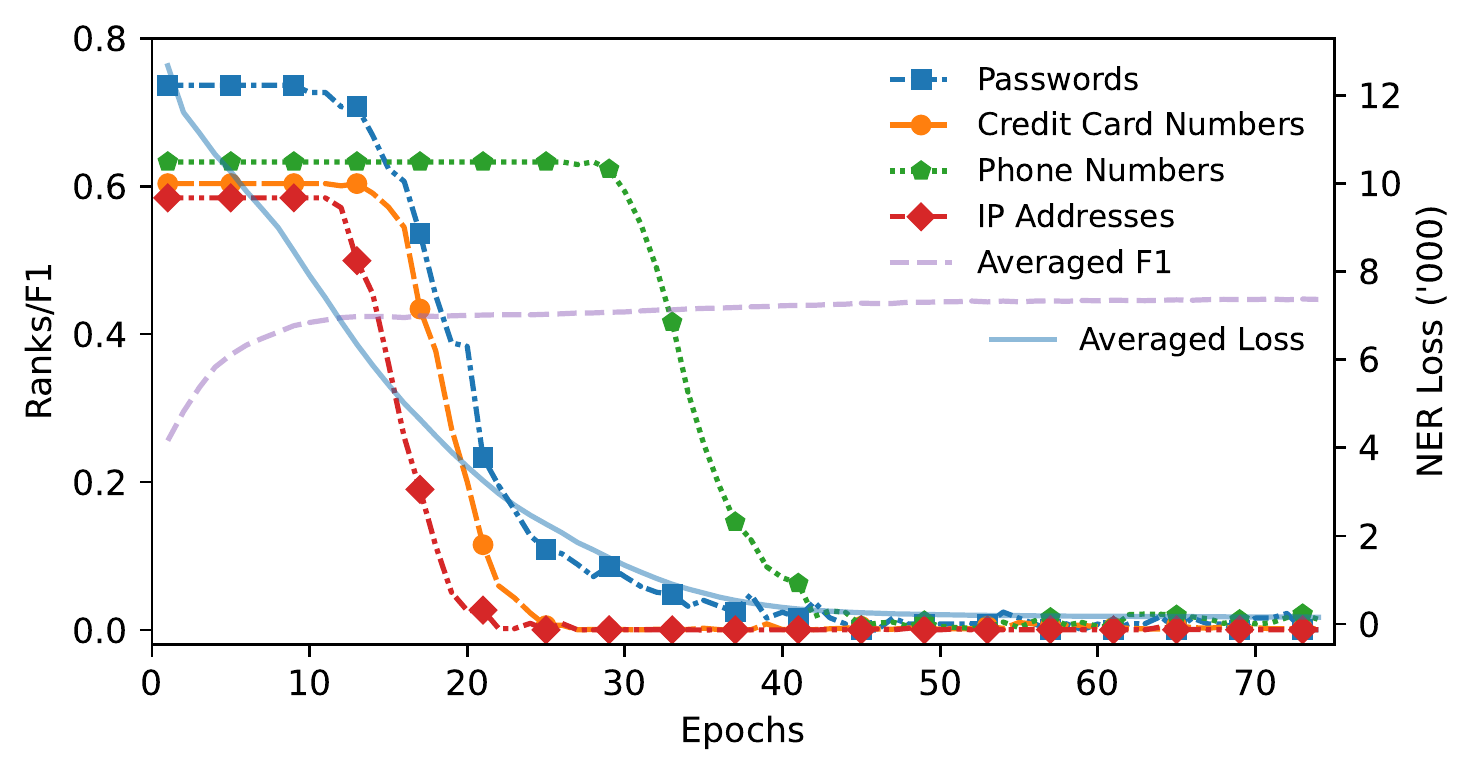}
\caption{With just a single phrase in a training set of 2500 examples, we demonstrate the phrase is still memorized over a short number of training epochs. We see a similar behaviour for different types of secret. Note that the loss line in the graph is overlapping for each of the secret type.}
\label{fig:rank_wrt_epoch_gum}
\end{figure}

\subsubsection{Effect of Similar Secrets}
The analysis in the preceding sections illustrates that if the target secret is within the set of secrets $\mathcal{R}$ available to the attacker, then it is possible to infer the secret via membership inference, i.e., the target secret being of rank 0. In a more realistic scenario, however, the target secret might not be in the set $\mathcal{R}$.
In this case, what does the highest rank secret in $\mathcal{R}$ tell us about the target secret? If the highest rank secret is close to the target secret, then the attacker has at least come close to the target secret. The attacker could then be able to infer the target secret by making small modifications to the candidate with the highest rank in order to ``hill climb'' to the target secret, as we shall see in the next section. 

Taking the password experiment as an example case, we analyzed the model's behaviour for secrets that look similar to the target secret. We first chose $300,000$ initial passwords drawn randomly from the public set of 1M passwords. We then chose $N = 10$ passwords randomly from this initial set. For each of the $N$ passwords, we constructed $70,000$ additional passwords based on three of the four sub-word features (Section~\ref{ner_spacy}) excluding the norm feature (lower case form of the password) as follows: 
we generate $10,000$ passwords each with (1) the same prefix only, (2) same suffix only, (3) same shape only, (4) same prefix and suffix, (5) same prefix and shape, (6) same suffix and shape, and finally, (7) same prefix, suffix and shape. For (1) (2) and (4) we assume we know the length of the password and add random alphanumeric characters for the rest of the password. For (3) (5) (6) and (7) we use the shape and fill  in missing characters using random alphanumeric characters depending on the characters defined the shape of the target password. These 7 different combinations of sub-word features result in $70,000$ passwords for each of the $N = 10$ passwords, resulting in a total of 700,000 additional passwords. Note that while generating these passwords, we assume that the length of the target password is known. We call these the \emph{feature passwords}, and denote the set by $S$. In total, including the 300,000 initial passwords, we have 1,000,000 passwords in our secret space $\mathcal{R}$. To calculate distance between passwords, we use the \textbf{f-distance} metric as defined in Eq.~\ref{eq:f-dist}. The reason for generating feature passwords is to have more passwords closer to the target password in terms of \textbf{f-distance}.

We then construct a phrase $s[r]$ for each of the $N$ passwords, update the NER model with this phrase only, and then query the updated NER model with all passwords in $\mathcal{R}$, and record the rank. We update the NER model over 75 epochs, and repeat this procedure 48 times to get an average rank for all passwords. We then normalize the rank between 0 and 1 and sort the results according to the \textbf{f-distance} from the target password. An \textbf{f-distance} of 0 means that it is the target password. As we can see in Figure~\ref{fig:feature_distance_ranks}, the average rank increases as we increase \textbf{f-distance}. This means that secrets (passwords) that are closer to the target secret (password) are ranked highest (closer to 0.0), and secrets further away from the target secret are ranked lower (closer to 1.0).  


\begin{figure}
\includegraphics[width=\linewidth]{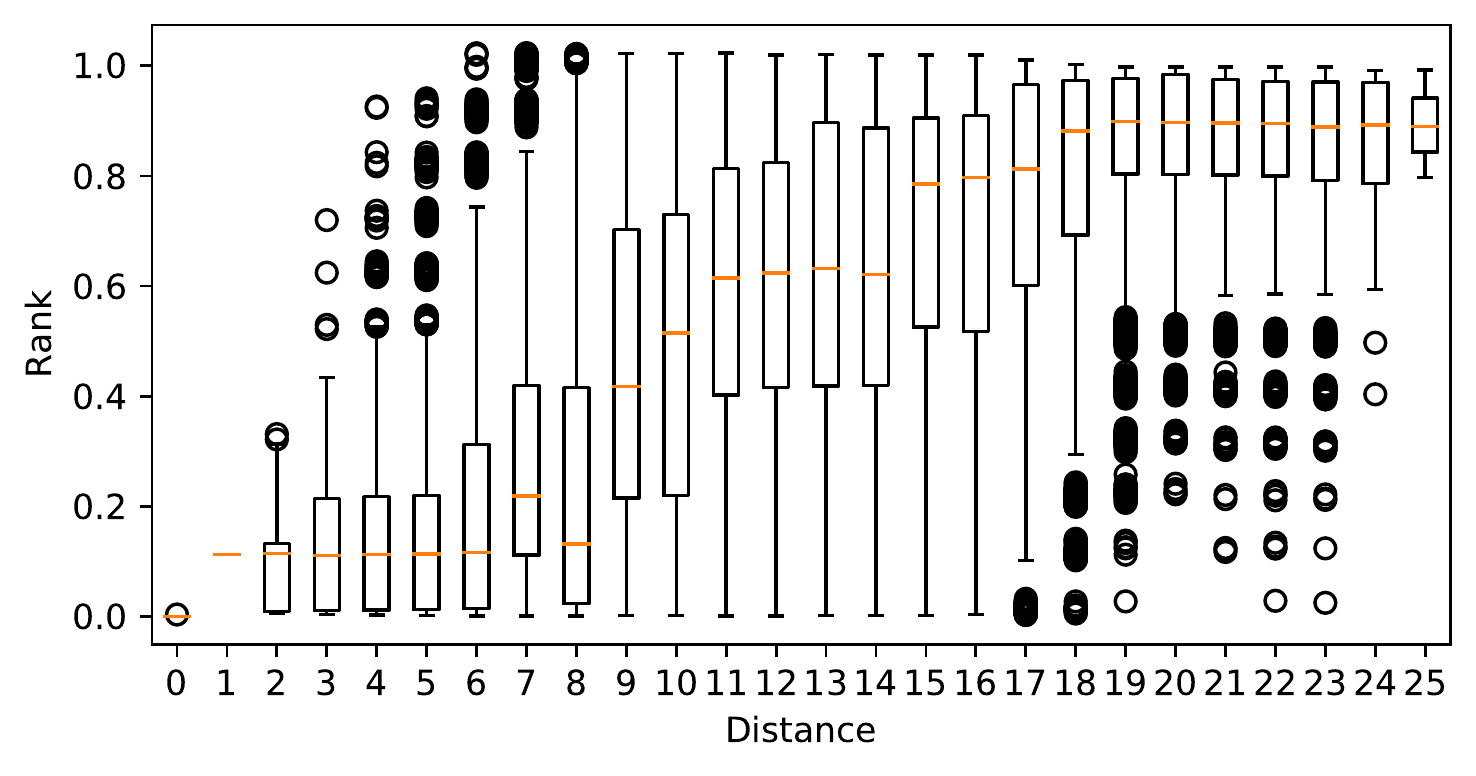}
\caption{The correlation between the ranks of the words with the \textbf{f-distance} they are away from the target word. }
\label{fig:feature_distance_ranks}
\end{figure}

%% file: experiment.tex
\section{The Document Redaction Use-Case}
\label{sec:doc-redact}
We now turn to demonstrating our attack on a more practical use-case. We assume a service provider that builds on top of spaCy's NER model to provide an online document redaction service. The service provider uses a medical dataset containing sensitive information such as patient IDs, contact information and place of residence. The NER model is updated by training on this dataset to learn new labels specific to the application of medical data redaction. We call this the \emph{updating dataset}. These labels include ID, CONTACT and LOCATION. The updated NER model can then be queried via an API which customers can use to redact information within their documents. We assume an attacker who has access to the same API, and wishes to infer sensitive information (related to the labels) contained in the updating dataset.

\subsection{The i2b2/UTHealth Corpus}
\label{subsec:MIA_medical_dataset}
We update the pre-trained spaCy NER model with the i2b2/UTHealth challenge dataset~\cite{stubbs2015annotating} which contains patient medical reports. These documents contain sensitive identifiers like names, IDs, contact, and location. We identify three sensitive labels (entities): ID, CONTACT and LOCATION for our experiments. We define a \emph{sample} for a label $l$ as any subsequence of words with the true label $l$ within these reports together with the surrounding sentence. We use the original training set (785 documents) and test set (511 documents) split from~\cite{stubbs2015annotating}. Samples from the training set are considered as members whereas those from the test set are non-members. Once again, we assume the attacker knows the surrounding sentences. Section~\ref{sub:unknown-surround} considers the case when the surrounding sentences may not be known by the attacker.

\descr{Attack Setup.} We use the membership inference attack proposed by Salem et al.~\cite{salem2019ml}, which uses the confidence score (probability score) returned by the model to infer membership. As mentioned before, the underlying assumption of this attack is that machine learning models may return a higher confidence score for training data points versus points seen by the model for the first time. In our case, the probability score is the score returned by spaCy when querying the model with a sample (of the three labels) from either the training or test dataset. We first train the model for 75 epochs with a batch size of 50 and evaluate the model's NER Test F1, and ROC Area Under the Curve (AUC) for membership inference after each epoch during training. Notice that since we have full training and test datasets, we can use the traditional AUC metric instead of the rank metric used in the previous section where we assessed membership inference of a single phrase.

\descr{Results.} Figure~\ref{fig:i2b2_mia} reports the AUC of the MI attack. Observe that there is a delay in the rise of the attack AUC for a few initial epochs due to the initial learning from the high number of samples for each entity in the training set. It is not until the test F1 score begins to plateau (training loss was also levelling off near 0) that we observe a sharp rise in the attack AUC above a random guess of 0.5. We also observe that the MI attack AUC for the ``LOCATION'' entity is higher than ``ID'' and ``CONTACT''. 
We hypothesize this is due to ``LOCATION'' words being linguistically common (and so the language model will contain knowledge of their semantics). This is also true, to a lesser extent of names in CONTACT entities. IDs, on the other hand is a string of numbers. Word embeddings are unlikely to include long numbers and thus unlikely to encode deeper semantics of numbers. However we note that the pre-trained NER model contains an existing ``LOCATION'' entity type. Contrary to our observation, one would expect a previously defined class in a pre-trained language model to already be sufficiently generalized to dissuade memorization. {The highest attack AUC is between 0.7 to 0.8, and even the lowest attack AUC is 0.6 (for the ID entity), well above the random guess mark of 0.5.} Therefore, we observe that MI succeeds in all cases. The relationship of MI accuracy between these entity classes is beyond the scope of this work.


Above we observed that ``ID'' and ``CONTACT'' had less documents and MI AUC. But to ensure a lower number of documents representing these entities in the training set is not responsible for the decreased MI AUC, we withheld documents containing ``CONTACT'' samples. Curiously, we observed the opposite effect for the same complexity. Specifically we performed an attempt by taking only documents where the ``CONTACT'' word was longer than 5 characters (yielding 94 samples), we refer to this as ``CONTACT-5''. By further restricting this set to 50, 30 and 10 documents, we also obtain configurations of ``CONTACT-5-50'', ``CONTACT-5-30'', ``CONTACT-5-10'' respectively. We observe with a decreasing number of available samples representing an entity, the attack AUC increases. Therefore, the less documents for an entity, the more the model will memorize the entity. Note that the F1-score approaching one indicates that in all cases the model is not prone to overfitting after around 10 epochs. We note that the F1 curves for all datasets is approximately the same, and thus have been averaged in Figure~\ref{fig:i2b2_mia}. The individual F1-scores of each experiment together with the detailed standard deviation bands are included in Figure~\ref{fig:full_i2b2_loss} of Appendix~\ref{app:full_docredact_loss}.




\begin{figure}[t]
\centering
\includegraphics[width=0.99\linewidth]{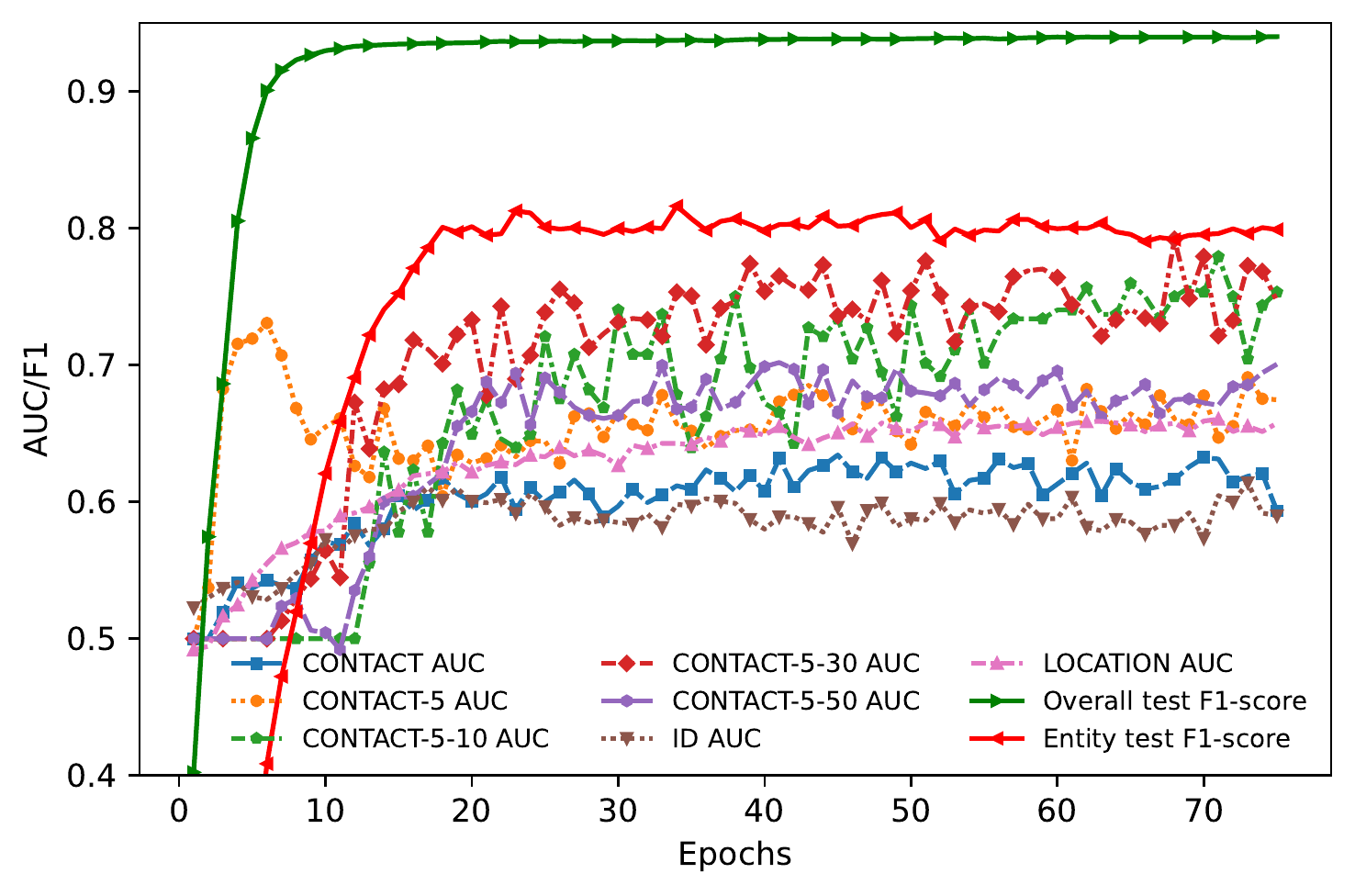}
\caption{{MI AUC reported over training time for different entities in the i2b2 dataset.}}
\label{fig:i2b2_mia}
\end{figure}





\subsection{Secret Redaction: Known Surrounding Sentences}
\label{subsubsec:i2b2_secret_redaction}
In Section~\ref{subsubsec:secret_types}, we observed that secrets such as passwords can still be inferred even if there are other labels in the private dataset with the example of the GUM corpus. To give more credence to this observation, we augment the i2b2/UTHealth dataset with phrases containing passwords. This mimics the scenario where the updated NER model is also responsible for redacting credentials for controlling access to medical records.
As before, we assume that the attacker is able to use the exact surrounding sentence on what may be a secret (password) to perform membership inference on a dictionary of candidate secrets.

To test this scenario, we sample two sets of 150 passwords to serve as our member and non-member set. We then craft a single phrase ``The secret phrase is ****'' to precede all passwords. The phrases with member passwords are then inserted into the i2b2/UTHealth training data to update the spaCy NER model. With both the member and non-member set, the AUC of MI is computed. This experiment was repeated reducing the passwords within the training set (150 to 100, 50 and 30). All experiments were repeated 10 times for an average result. We also report the testing F1-score to evaluate the performance of the trained model (and note that no overfitting has been observed).



\descr{Results.} From Figure~\ref{fig:same_phrases_main}, we observe that the MI attacker can successfully infer membership of passwords. Furthermore, the ability to infer passwords increases with smaller number of examples in the training dataset (SECRET-30 vs SECRET-150), which is consistent with the observation around the CONTACT label in Figure~\ref{fig:i2b2_mia}. We also observe that the F1-score for the whole model (on the medical dataset) as well as on the SECRET label is high (close to 1.0). This is in contrast with the GUM corpus (Section~\ref{subsubsec:secret_types}, Figure~\ref{fig:rank_wrt_epoch_gum}) where a high F1-score was illusive. This shows that memorization is occurring even when the model is not overfitted. The individualized MI AUC, NER training loss and individual F1-score with standard deviation bands of the different entities in Figure~\ref{fig:same_phrases_main} are shown in Figure~\ref{fig:full_i2b2_secret_loss} in Appendix~\ref{app:full_docredact_loss}.    


\subsection{Password Redaction: Unknown Surrounding Sentences}
\label{sub:unknown-surround}
Until now we have assumed that the attacker knows the surrounding sentences of the target secrets. We consider the setting where the attacker does not know the exact surrounding sentences used in the private dataset. Note that this is no longer membership inference, but rather an attribute inference attack. 
We achieve this by repeating the experiment from Section~\ref{subsubsec:i2b2_secret_redaction}, but instead of a single phrase preceding the secret, we craft four distinct phrases for passwords in the member set: (``My secret is ****'', ``The password to my account is ****'', ``Here is the password: ****'', ``Log in using the password: ****''). Each password in the training dataset is randomly assigned to one of the four phrases. On the other hand, the non-member passwords have a different phrase ``The secret phrase is ****''. We also assume that the attacker uses this phrase as the surrounding sentence.

\begin{figure*}[t]
\centering
\begin{subfigure}{0.49\linewidth}
\centering
\includegraphics[width=0.99\linewidth]{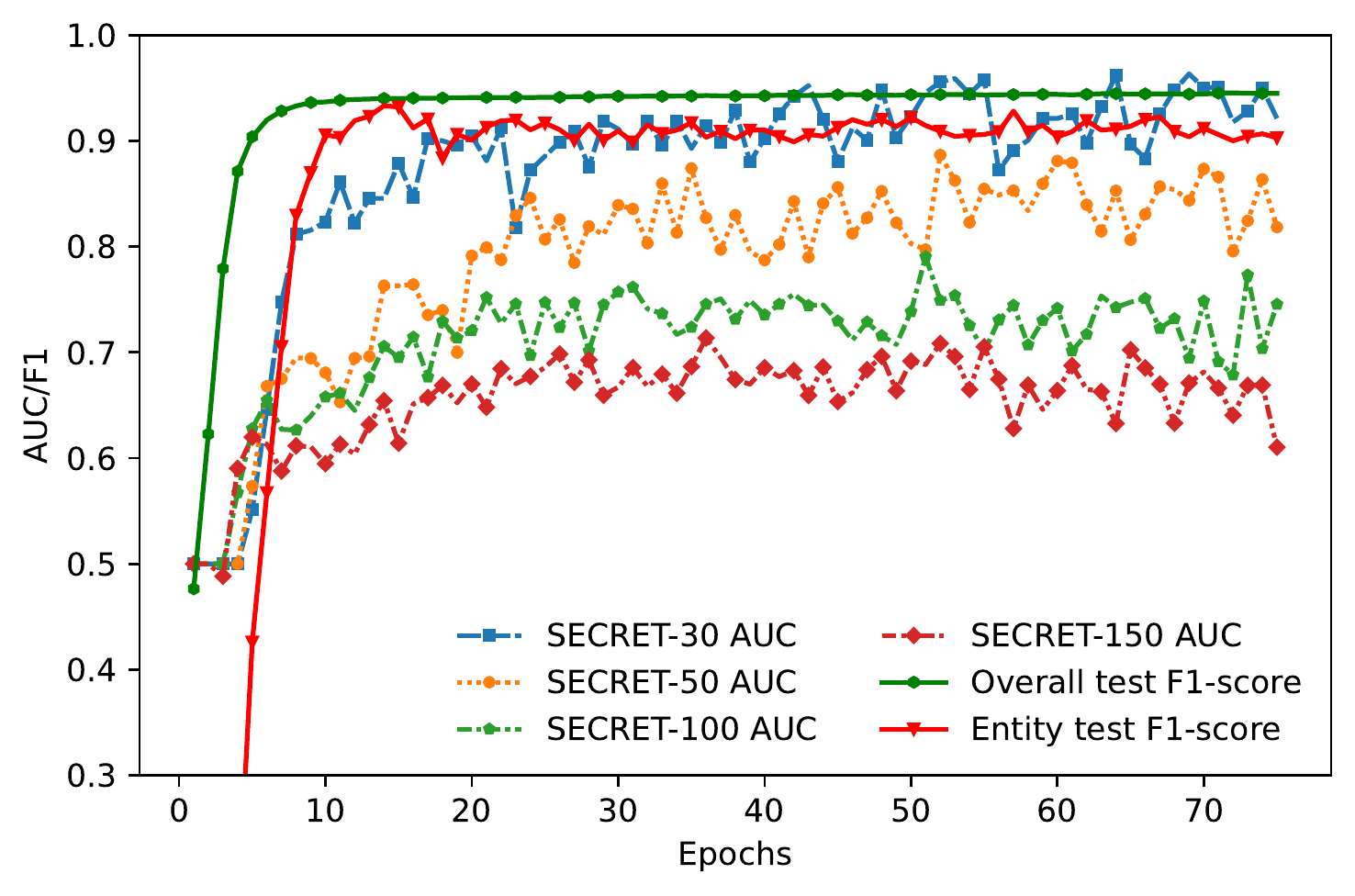}
\caption{{Same surrounding sentence as training.}}
\label{fig:same_phrases_main}
\end{subfigure}
\begin{subfigure}{0.49\linewidth}
\centering
\includegraphics[width=0.99\linewidth]{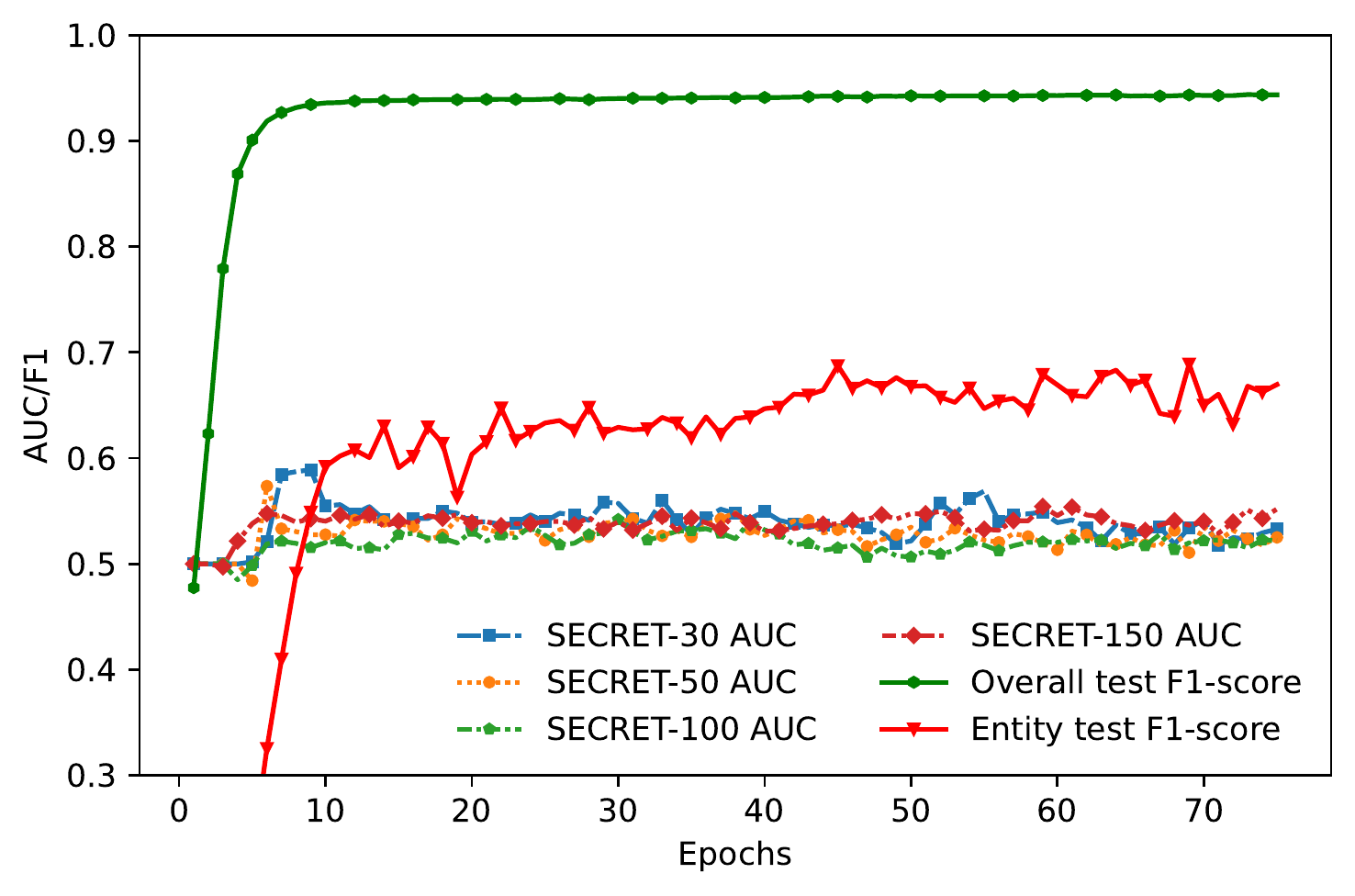}
\caption{
{Different surrounding sentences than training.}}
\label{fig:diff_phrases_main}
\end{subfigure}
\caption{MI AUC for each training epoch with differing numbers of training secrets inserted into the i2b2 dataset. Potential secrets are queried with surrounding sentences. The overall F1 score was computed on testing data. }
\end{figure*}


\descr{Results.} From Figure~\ref{fig:diff_phrases_main}, we observe that the performance of MI on unseen surrounding sentences is close to a random guess (0.5 AUC). This observation is not significantly influenced by the number of samples present in training the secret class in contrast to prior results. However, we note that the F1 score of redacting secrets has degraded to 0.65, with the use of different surrounding sentences in the test set to those in the training dataset (c.f. Figure~\ref{fig:same_phrases_main}), irrespective of which secrets were used.

The performance of the NER model is decreased if it has not previously seen a surrounding sentence during training. Thus, the training process should contain sufficiently many representations of expected surrounding sentences for better accuracy. On the contrary, the attack is more successful if the attacker can craft surrounding sentences that are similar to the ones used in the private dataset. For instance, an attacker may employ a strategy to recover a sufficiently similar surrounding sentence (by checking if the redaction accuracy of the updated NER model increases on a generic secret), prior to performing the membership inference attack on the target secret. For instance, for the case of phone numbers, the attacker could craft a phone number distributed similar to the original data (e.g., a US mobile phone number). The attacker can then query multiple manually crafted sentences with this phone number, e.g., ``Please call me on ..., My phone number is ..., ... is where you can reach me.'' With each subsequent query, the surrounding sentence with the highest probability can be further perturbed by replacing words that increase the confidence of the whole sentence. The idea being that the manually crafted surrounding sentence closest to the actual surrounding sentence would give the highest confidence by virtue of being in the training dataset. Figure~\ref{fig:same_phrases_main} shows the plausibility of this approach. 
However, we defer detailed investigation of the impact of similar surrounding sentences as future work.

%% file: timing_side_channel.tex
\section{Membership Inference via Timing Attack}\label{sec:MIA_timing}
In this section, we propose a timing side-channel attack on NER models from observations that models process traning dataset words faster. The attack in prior sections was specific to memorization in the NN model, this attack targets the tokenizer which internally accesses a vocabulary maintained by the NER pipeline.
Like that of Section~\ref{threat_model}, the updated NER model is available to the attacker via black-box access, with an attacker objective to infer membership of sensitive training data used to update the model. An important distinction of this MI attack, is that surrounding sentences are inessential, and thus knowledge of the surrounding sentence is not assumed. 

\subsection{Timing attack methodology}

\subsubsection{Vocabulary}
A pretrained NER pipeline contains many components (We specifically evaluate spaCy's pre-trained \newline ``en\_core\_web\_lg'' and  ``en\_core\_web\_trf'' pipeline~\cite{spacy}, the key difference being the tokenizer, of which ``tok2vec'' and a ``ROBERTA based transformer'' is used respectively), each component can be treated as a separate model to be updated with user data. 
In these pipelines, there exists a vocabulary containing word associated strings (prefix, suffix, shape and norm) of the pre-training dataset. We call this the \textit{original vocabulary}. When one updates the pipeline with their private dataset, the vocabulary of the pipeline is updated with strings (prefix, suffix, shape and norm/lowercase for tok2vec, and byte-level byte pair encodings for the ROBERTA-base transformer) from the updating dataset (The \textit{updated vocabulary}), becoming a superset of the original vocabulary. Words within a vocabulary are in-vocab words, and those not are out-vocab words.

\subsubsection{Intuition Behind the Attack} 
The exploit occurs when the model takes additional time processing out-vocab words. Specifically, when a word passes through the pipeline, the tokenizer first extracts the tokens of the word which are checked in the vocabulary. If present, nothing changes, if not present, they are added to the vocabulary. Thus, vocabulary updates on new out-vocab words are likely to take longer. 
%
%
Member words from the training/update set are already present in the vocabulary, in contrast to unseen non-member words. This, presents an opportunity for an attacker to infer membership based on the inference time. 

For this attack, position ourselves as a black box attacker, leveraging only the execution time from the time of query to the time the output is returned. The assigned label and confidence of the word is ignored (Though there remains a potential to fuse these features). We demonstrate this attack with a password use-case, with an updated NER model is specifically trained to label passwords in documents. Again, the adversary's goal is to infer membership of passwords used in the updating (private) dataset.

\subsubsection{Attack on the Original NER Model}
\label{subsec:attack_orig_ner}
We first evaluate the timing side-channel attack on two public spaCy NER models. 
By accessing the pretrained vocabulary strings, we can establish the ground truth for members, or \textit{in-vocab set}. We remark that in practice, a black box attacker would not have simple access to the model vocabulary. Next, we generate a set of out-vocab words such that at minimum their `suffix' and `shape' are not in the original vocabulary's string table (the same set is used to evaluate the transformer pipeline). We note that we exclude the prefix feature when generating out-vocab words due to unavoidable overlaps with in-vocab words (by default, the first character of the word). 
Specifically, we generate a list of 5000 random passwords where each password has minimum length of 6 characters including at least 1 of each: uppercase, lowercase character, digit and special symbol. We refer to these generated passwords as the \textit{out-vocab set}.

\descr{Experimental Setup.} We randomly sample 2000 words from the \textit{in-vocab set} and 2000 password from the \textit{out-vocab set}. Then, we sequentially query the original NER model interleaving between in-vocab and out-vocab words. We alternate to ensure consistent delay across both sets as a result of other spurious processes on the system. Each query is simply a single word with no surrounding text, with a measurement of the time taken to process the query. 
From the execution time, the attacker can choose a threshold to classify the queried word as a member or non-member. For our evaluation we present the ROC curve and the AUC as a general measure of classification performance.

\subsubsection{Attack on an Updated NER Model}
\label{timing_attack_updated_ner}
We also evaluate an updated version of the NER model. Starting with the clean pretrained NER models, we update them with a subset of our generated passwords. 
To update the NER model, we generate 5000 passwords like Section~\ref{subsec:attack_orig_ner}, reserving the first 2000 passwords (\textit{updating dataset $D_u$}) to update the NER model. 
The NER model update includes teaching the NER to identify sensitive words with the new SECRET label. Specifically, training examples are in the form $s[r] = \text{``Alice's secret is $r$.''}$, where $r$ is a password from $D_u$, annotated as:
\begin{quote}
   ``Alice's secret is $r$.'',
\{`entities': [(0, 5, `PERSON'), (18, 18 + length($r$), `SECRET')]\}
\end{quote}

\descr{Experimental Setup.} We sample 2000 passwords from the remaining 3000 generated passwords. Then, we query the updated NER model with 2000 $r\in D_u$ (in-vocab) and 2000 $r\notin D_u$ (out-vocab), in an alternating manner. The execution time to output return for each query is measured. 
Again, the attacker can choose a threshold to classify the queried word as a member or non-member, though we present the ROC curve and the AUC as a general measure of performance.

\subsection{Timing Attack Results}

\descr{Timing Attack.} 
Figure~\ref{fig:timing_sid_channel} shows the receiver operating characteristic (ROC) curve and the AUC values of our attack on the original NER and the updated NER models.
We observe significant ease in inferring membership of words in both the original the updated NER pipelines. The lowest AUC is 0.9923, and is consistent across both the ``tok2vec'' and the ROBERTA based transformer pipelines. In both pipelines, the updated NER is marginally more susceptible to MI. It is clear that timing attacks against these pipelines are effective

We provide an expanded view of the in-vocab and out-vocab times in Figure~\ref{fig:absolute_time_ner} of Appendix~\ref{sec:app:timing}. From those figures, it can be observed in every setting, the NER model takes more time to process words that are outside the pre-trained vocabulary (in the original NER model) and outside of the updating dataset (in case of the updated NER model).

\begin{figure}[t]
\centering
\includegraphics[width=0.95\linewidth]{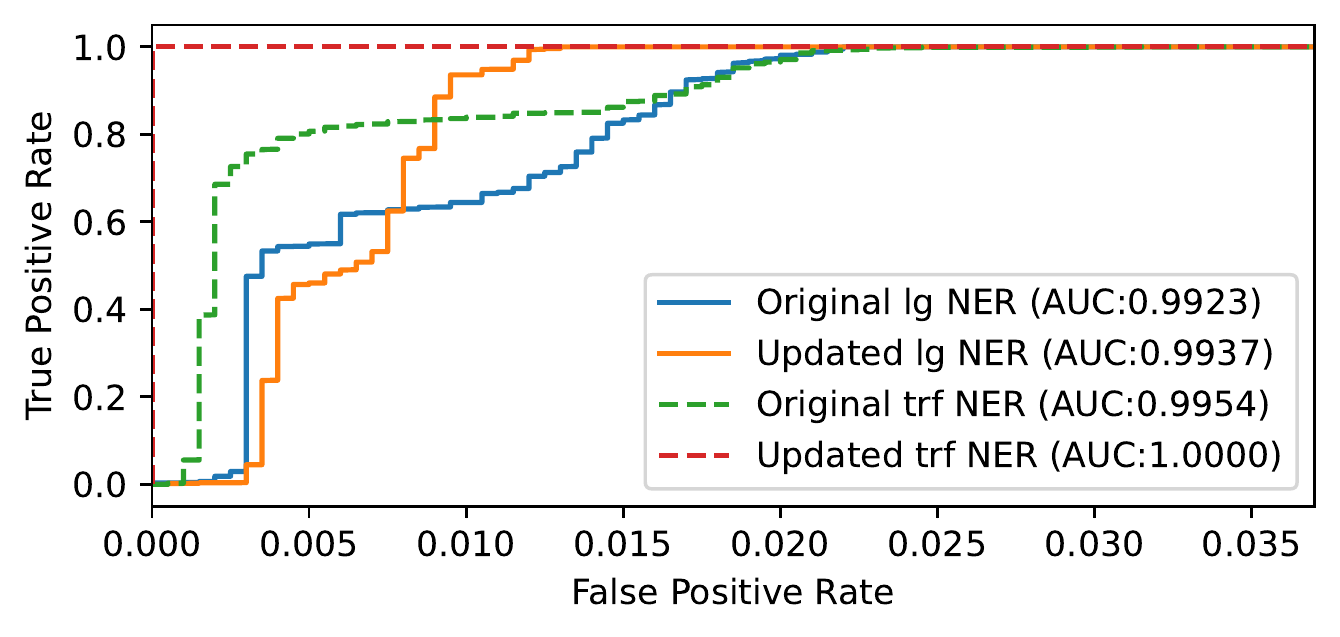}
\caption{Timing side-channel attack ROC curve to classify membership of word inputs for the original pre-trained and updated NER model, on both tok2vec and transformer pipelines. {Note that the x-axis does not span to 1.}
}
\label{fig:timing_sid_channel}
\end{figure}


\descr{Querying Non-Members Multiple Times.}  
Following the attack we observe that the features of even non-member words will be added to the vocabulary of the model. Noting the fact that spaCy's NER pipeline will add features of any new word into the vocabulary, we continue our attack on both versions of the NER model where an adversary queries the model with the same word (in-vocab and out-vocab) two more times. Surprisingly, we still obtain AUCs of greater than 0.8195 for both models and pipelines with the execution time of the model to process words on its \textit{second} or \textit{third} query (In full, 0.820, 0.994, 0.994, 1.00, for 2nd run, and 0.995, 0.995, 0.993, 0.999 on the 3rd run for Original LG, Updated LG, Original TRF, and Updated TRF respectively). 
{The persistence of timing vulnerability despite the vocabulary update means that there are additional components within the processing pipeline that are not updated, resulting in sufficient timing differences for successful membership inference. Regardless, we can conclude that NER models take more time to process words they are not trained on.} 



%% file: defences.tex
\section{Possible Defences}\label{sec:defences}

\descr{Neural Memorization.} 
Memorizing training data poses serious privacy risk in machine learning. It is crucial to measure how much memorization occurs and propose solutions to mitigate this issue or at least to train the model with less memorization. 
Several possible approaches to mitigate unintended memorization of NNs in text generation were discussed in~\cite{carlini2019secret}. Among them, using differential privacy is demonstrated to be effective as it could fully eliminate unintended memorization in text predictive models~\cite[Sec 9.3]{carlini2019secret}. 
An emergent relationship between the number of unique documents in an entity class and the success of the membership inference attack was observed in Section~\ref{subsec:MIA_medical_dataset}. From this relationship, a simple means to reducing (but not eliminating) the success of a membership inference attack is to either eliminate the entity, likely an infeasible option as the capability to redact these infrequent entities is destroyed. Or alternatively, to generate, or source additional similar entities to train the model, to populate and reduce the risk posed by the under-represented entity class.
An approach to mitigate the proposed attack exploiting model's memorization demonstrated in Section~\ref{subsec:MIA_medical_dataset} is to restrict the NER model to return only the label of the queried text instead of without the confidence scores. Though we note that Label-only MI attacks have since been developed~\cite{li2020label,choo2020label}, we leave the evaluation of these label-only attacks as future work.

\descr{Timing Attack.} 
A possible defence to this vulnerability is to disable the updating vocabulary step during model training (or updating) so that sensitive words included in the private dataset are not added to the vocabulary. It is also noted that adding every unseen word's features to the vocabulary might make the model hosted in a cloud to be vulnerable to attacks which query the model with drivel, i.e., inundating the vocabulary with bogus words which could lead to space issues.  
 
Masking timing characteristics has been proposed to prevent against timing attacks in cryptosystems~\cite{kocher2018timing}. It is possible to apply this idea to prevent timing side-channel in NER models. {For instance, one} approach to mitigate the timing side-channel attack is to delay returning the result until a pre-defined time has elapsed. This would remove any timing difference, but impacts overall system speed. {Another approach specific to our scenario is to add a constant delay to the time of the member words.} For instance, on the updated LG NER model (cf. Section~\ref{timing_attack_updated_ner}), the average execution time taken by the model to process the member words and non-member words are $t_{\text{m}} = 2.08 \text{ms}$ and $t_{\text{nm}} = 2.44 \text{ms}$, respectively. {Assuming both time distributions to be normally distributed with means $t_{\text{m}}$ and $t_{\text{nm}}$, and the same variance, adding a time of $t_\text{nm} - t_\text{m}$ to the time of member words will make the two distributions identical, and hence thwart the timing attack. In practice system delays, or contested resources may result in a few outliers, in which case a trimmed mean can be employed. We note that the success of this defense can only be empirically determined as it relies on the accuracy of estimated mean times of member and non-member words, how accurately the distribution of times can be approximated as normally distributed, and the threshold to determine outliers. 
} 

{
\paragraph{Experiments:} In order to test the constant delay defense, we added the constant time of $t_\text{nm} - t_\text{m}$ to member queries in the experiments of Section~\ref{timing_attack_updated_ner}. To remove outliers, the we calculated a 10\% trimmed mean. For instance, a $t_\text{nm} - t_\text{m}= 36 \text{ms}$ delay is applied to the updated LG NER model. From Figure~\ref{fig:timing_sid_channel_def}, it is immediately evident that the timing attack can be thwarted with the addition of constant time. In all experimental settings, the MI AUC is close to 0.5, or a random guess.  The slight deviation from 0.5 is due to the fact that the timing distribution is not exactly standard normal (see Figure~\ref{fig:absolute_time_ner} in Appendix~\ref{sec:app:timing}). Compare this to Figure~\ref{fig:timing_sid_channel}, which depicts the MI ROC with no defense, where AUCs are close to 1.0. 
}

\begin{figure}[t]
\centering
\includegraphics[width=0.95\linewidth]{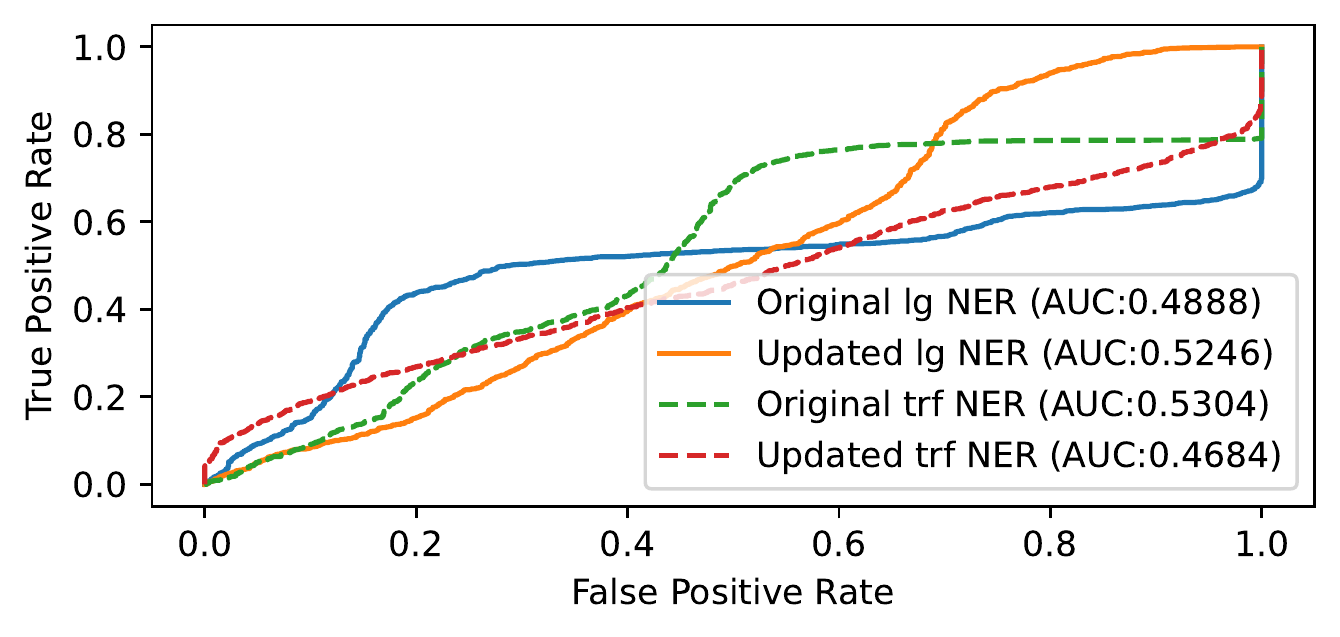}
\caption{{Constant time delay added to member words to defend against the timing side-channel attack. ROC curve classifies membership of word inputs.}
}
\label{fig:timing_sid_channel_def}
\end{figure}


\descr{Using Synthetic Data.} A potential defense to thwart MI attempts on either attack is replacing sensitive information in the private dataset with dummy information. However, this is not straightforward, as the owner will need to ensure that the dummy information at least preserves the semantics of the original information without degrading the accuracy of the NER model, e.g., length of ZIP codes, format of locations, and IP addresses, albeit in some cases, e.g., passwords, this may be possible. Moreover, this can be a tedious task for a large dataset, and some sensitive information may inadvertently be missed. Our attacks highlight the dangers of using sensitive datasets for training NER models without necessary protections. 
{
\paragraph{Experiments:} To test this mitigation measure we perform additional experiments with the i2b2 dataset. The i2b2 dataset is prepared by creating sets of members and non-members as in Section~\ref{subsubsec:i2b2_secret_redaction}. However, as an additional step, we replace instances of the sensitive label with dummy information from all training members. In particular, we choose the sensitive labels of Location and Contact, and replace all addresses and phone numbers with publicly available randomly generated synthetic ones.\footnote{\url{www.randomlists.com}} The model is then trained with this dummy training set. Note that the surrounding sentences remain the same. To test MIA accuracy we use the non-member set records from the same distribution as the original i2b2 dataset. We perform this experiment for two of the potentially vulnerable entities, i.e., Location and Contact. 

The results illustrated in Figure~\ref{fig:dummy_location} and \ref{fig:dummy_contact} show that the performance of the MI attack has now dropped to an AUC comparable to a random guess, demonstrating that a model trained with dummy data is less susceptible to membership inference attacks. This however, does result in a significant degradation of the model's ability to identify the correct entities due to the difference in the distribution of the dummy and original dataset. To resolve this trade-off, the task is to generate synthetic data such that the dummy data matches the original dataset's distribution, without compromising the privacy of the sensitive field. This means carefully crafting synthetic examples that capture sub-categories of the entities, e.g., the LOCATION category may contain different granularity of address information: from city, to zip code, to street address. This cannot be captured by simple data generation techniques such as the one used by us. We however leave this consideration for future work, given the setting and domain specificity of each sensitive entity. Figure~\ref{fig:defend_i2b2} in Appendix~\ref{app:full_docredact_loss}, shows the overall F1-score for all entities together with the standard deviation bands for these two experiments.}

\begin{figure}[t]
\centering
\begin{subfigure}{\linewidth}
\centering
\includegraphics[width=0.99\linewidth]{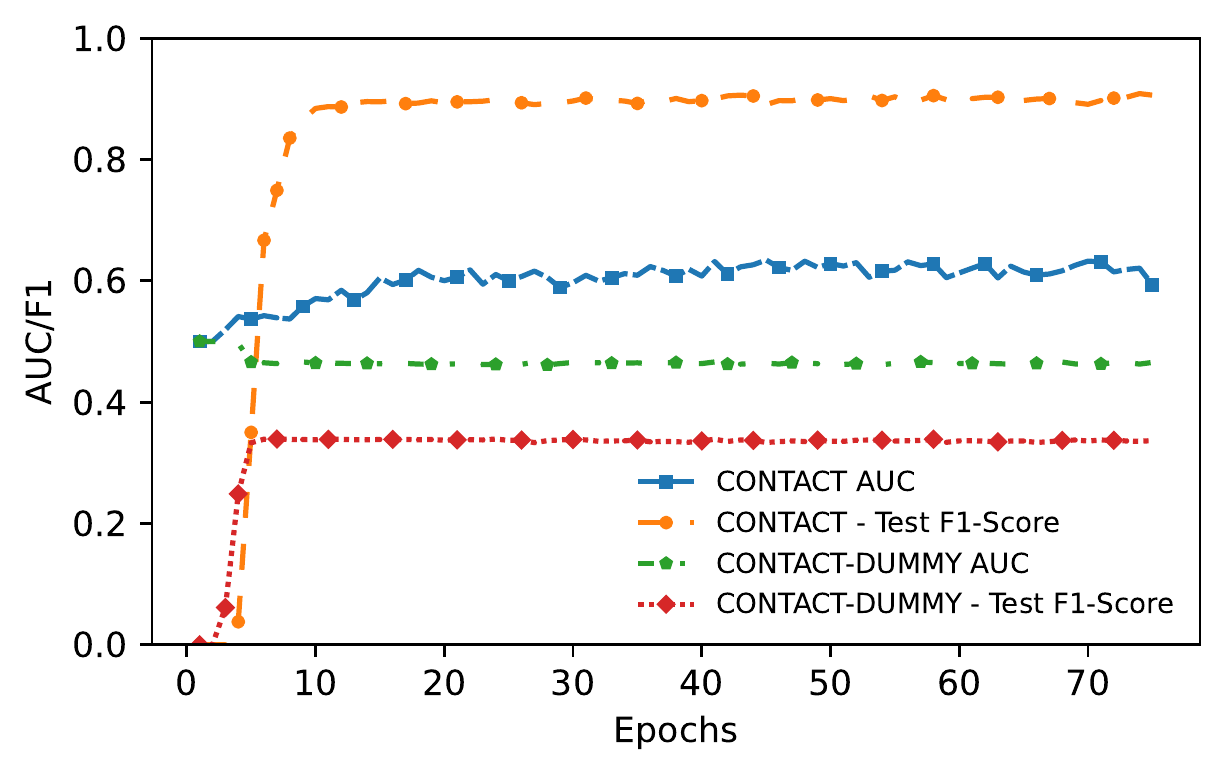}
\caption{
{Original vs DUMMY LOCATION MI AUC}
}
\label{fig:dummy_location}
\end{subfigure}
\begin{subfigure}{\linewidth}
\centering
\includegraphics[width=0.99\linewidth]{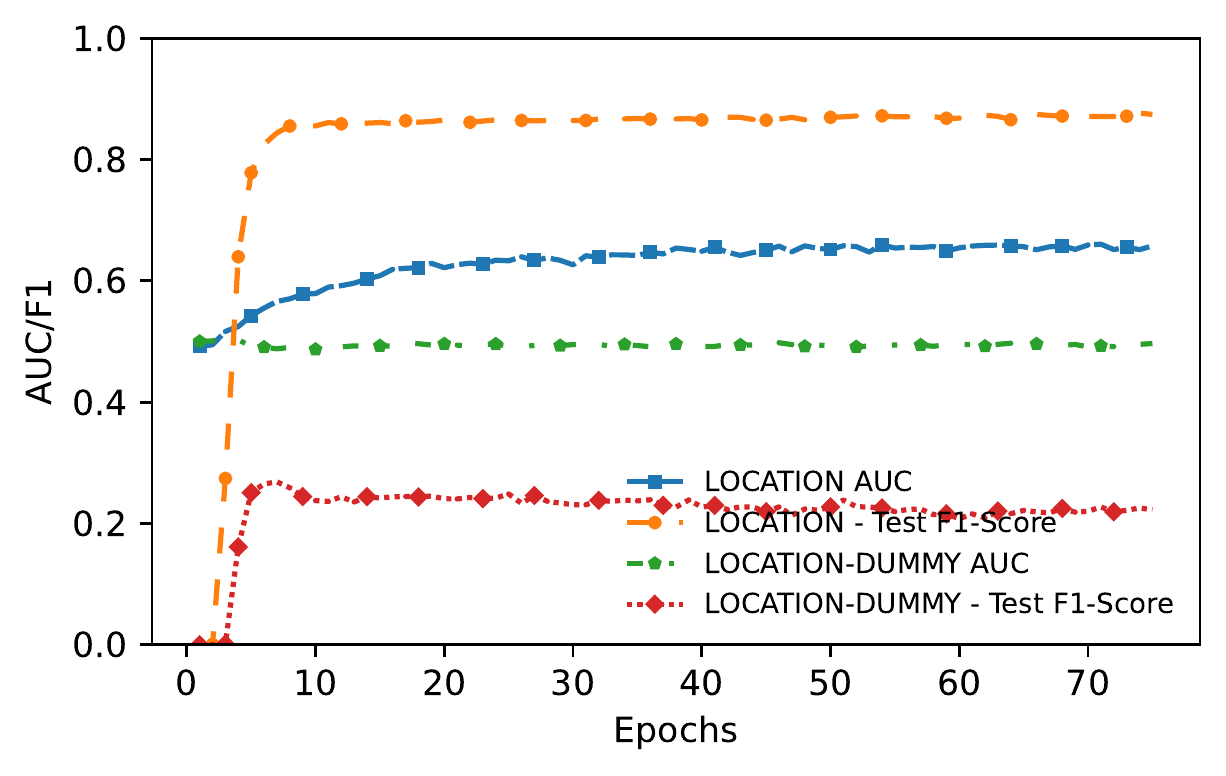}
\caption{
{Original vs DUMMY CONTACT MI AUC}
}
\label{fig:dummy_contact}
\end{subfigure}
\caption{{MI AUC comparison between Original dataset entity from Figure \ref{fig:i2b2_mia} vs DUMMY data for CONTACT and LOCATION.}}
\end{figure}

%% file: related_work.tex
\section{Related Work}\label{sec:related_work}



\textit{Attacks against Text Classification Models.} Recent works have demonstrated a vulnerability of text classification models to adversarial sample attacks, reducing model accuracy or forcing the output of a target label~\cite{liang2017deep, jin2020bert, garg-ramakrishnan-2020-bae}. Garg et al.~\cite{garg-ramakrishnan-2020-bae} propose BERT-based adversarial examples with contextual perturbations from the BERT masked language model, to reduce classifier accuracy by more than 80\%. In contrast, we focus on analyzing the vulnerability of text classification models to membership inference and timing side-channel attacks.

\textit{Membership Inference (MI) Attacks.} With access to a ML model an attacker can determine if a data sample was in the training dataset, or infer its membership, Shokri et al.~\cite{shokri2017membership} first achieved this by training shadow models to mimic the behavior of the target model. Salem et al.~\cite{salem2019ml} proposed a low cost attack on ML models without the need of shadow models, in addition to a relaxed attack that sets a threshold for prediction confidence to identify membership. 
On text generation models, Song et al.~\cite{song2019auditing} present a user-level MI attack by training an auditing model that checks the membership of any user sample even if not explicitly used in training the model.
Salem et al.~\cite{salem2020updates} infers labels of updating data samples and reconstructs the updating data based on output differences between two versions of models. Moreover, this work targets image classification where as our focus is text classification.

\textit{MI Attacks against NER.}
To the best of our knowledge the only work analyzing MI attacks on Named Entity Recognition is~\cite{seyedi2021analysis}. Specifically, Seyedi et al. perform \cite{shokri2017membership}'s Shadow MI attack on a NER model (NeuroNER~\cite{NeuroNER}) used to de-identify medical records. However, Seyedi et al. concludes an inability to distinguish between members and non-members of the medical training data.  
However, as previously seen in Section~\ref{subsec:MIA_medical_dataset}, contrary to Seyedi et al.'s conclusion, we find there exists a realistic MI threat against the i2b2 medical dataset~\cite{stubbs2015annotating} and when applying MI on sensitive information like passwords.
Furthermore, the NeuroNER model contains a tokenizer, opening potential vulnerabilities to our timing side-channel attack. 

\textit{Memorization in Language Models.} Carlini et al.~\cite{carlini2019secret} demonstrate generative text models memorize out-of-distribution secret phrases in the training data and present methods to extract these secrets.
\cite{carlini2020extracting} has recently shown that large language models like GPT-2, trained on text from the internet including sensitive information from billions of users can be attacked to extract personally identifiable information (PII), in extremely large parameter models. Despite these large models demonstrating no memorization through generality tests, they still are capable of leaking information.  
Helali et al.~\cite{helali2020assessing} demonstrates the unintended memorization of NER models specific to labels that share linguistic materials with at least one other label. In this work, we demonstrate that text classification models, specifically Named Entity Recognition models also unintentionally memorize their training data with proposed attacks that exploit said memorization to infer the membership of sensitive information.
Zanella et al.~\cite{zanella2020analyzing} uses differential scores and ranks from two model versions to analyze information leakage of an update set for text prediction tasks.

\textit{Timing side-channel Attacks.}
A timing side-channel is defined as a side-channel that leaks information due to faulty or vulnerable operations of a system~\cite{biswas2017survey}. In literature, numerous timing side-channel attacks have been proposed to attack cryptographic systems. 
For instance, attacks may exploit delays in memory access of a program due to other programs accessing the same shared memory. Wang et al.~\cite{wang2014timing} proposed an attack to compromise the private key of RSA decryption by monitoring memory accesses. Wang et al.'s attacker continuously issues memory requests to memory shared by the RSA decryption program, measuring the end-to-end delay for request completion. With the delay, the attacker can estimate the number of `1' bits in the RSA's private key. Timing side-channel may also be based on execution time. Chen et al.~\cite{chen2013improving} proposed a timing side-channel attack against the RSA-CRT algorithm that executes a chosen ciphertext attack to collect time measurements on the overall decryption process.  
In neural networks, few works consider the timing side-channel. For instance, Duddu et al.\cite{duddu2018stealing} proposed a timing side-channel attack to reverse the neural network architecture. Specifically, the attacker queries the target model and measures the execution time which is then used to infer the model's depth. 
To the best of our knowledge, our work is the first that exploits the vulnerability of NER models to a timing side-channel for MI attacks.

{
\textit{Contextualizing Differences from Previous Work.} In this work we have shown that NER models are susceptible to membership and attribute inference attacks by exploiting neural memorization and timing side channels. As indicated above, neural memorization has already been demonstrated in the literature~\cite{carlini2019secret, carlini2020extracting}, however their setting is text generation instead of text classification. In the text generation setting, the model returns a ranking of possible word tokens given a sequence of words, which it has learned through its training. Thus, given a sequence such as ``The secret phrase is ...,'' the model returns possible tokens that would complete the sentence. If the model is susceptible to neural memorization, the sensitive word (password in this case) would be higher ranked in the output. From this, a membership inference attack can be launched in a straightforward way. Our text classification setting is different, as the model never suggests tokens, and we need to modify the attack to try different possibilities (e.g., password dictionaries) to guess the missing word (e.g., password).  

We also remark that membership and attribute inference attacks in general need not be related to neural memorization. For instance, membership inference attacks have been demonstrated on machine learning models other than neural networks~\cite{shokri2017membership}. Furthermore, membership inference attacks may use a different exploit other than neural memorization. One such example is overfitting due to over training which is different from neural memorization~\cite{carlini2019secret}. In this paper we have also shown how membership inference attacks can be launched via exploiting side channels, specifically, timing side channels. Note that timing side channels have previously only been demonstrated on stealing machine learning models (model parameters), and not on membership inference on natural language processing models~\cite{duddu2018stealing}. We have shown how this can be done by exploiting the built-in dictionary of member words in NER systems.
}

%% file: conclusion.tex
\section{Conclusion}\label{sec:conclusion}
Many emerging applications of Named Entity Recognition (NER) models show their potential to be monetized as a service, including identification of medical and other personal information in documents to minimise impacts of data breaches, and redaction of sensitive information for document sharing. With the spaCy NER model as a use case, we have shown that the current crop of neural network based NER models are vulnerable to privacy attacks on their underlying training data, and thus in their current form they are not suitable to be hosted as an online service for labeling sensitive information. We have demonstrated two main attacks on NER models, one based on neural network memorisation, and the other based on timing differences when queried with text appearing in the training dataset versus text outside the training dataset. Both attacks can be somewhat mitigated by using techniques such as differential privacy and timing delays, with obvious accuracy and efficiency tradeoffs. We hope our work motivates the pursuit for more privacy-preserving NER models.

%% file: appendix.tex
\section{Equivalence of Words in spaCy}
\label{app:equiv-proof}
We define the Hamming difference between two words $w_1$ and $w_2$ of equal length as the number of positions where they differ. That is, all positions $j$ such that $w_1(j) \neq w_2(j)$. Note that we consider the lower and upper case characters as two different characters. Two words $w_1$ and $w_2$ are considered the same, denoted $w_1 = w_2$, if and only if they are of the same length and their Hamming difference is zero. Otherwise they are considered different. We assume the prefix to be the first character and the suffix to be the last 3 characters; the default setting in spaCy as mentioned in Section~\ref{ner_spacy}. We say two words are equal in spaCy if they have the same prefix, suffix, norm and shape, otherwise they are different in spaCy. Our first observation is that two different words may be equal in spaCy. This is easy to see with the words: ``ABCDEfghij'' and ``ABCDefghij''. Both words have the same prefix, suffix, norm and shape (i.e., XXXXxxxx, since same character sequences are truncated after a length of 4). Thus, this separates the two notions of equality of words. We are interested in knowing the conditions under which the two notions differ.

For any word $w$, we let $\text{shape}(w)$ denote its shape under spaCy (which is also a string). We define the middle of the word $w$ as the substring obtained after removing its prefix and suffix. We denote it by $\text{mid}(w)$. 

\begin{proposition}
\label{prop:length}
Let $w_1$ and $w_2$ be two words equal in spaCy. Then their prefixes, suffixes and length are the same.
\end{proposition}
\begin{proof}
The fact that the prefixes and suffixes are the same follows from the fact that they are case-sensitive in spaCy. If $|w_1| \neq |w_2|$ then the two case-insensitive norms of the two words will be different, contradicting the assumption. 
\end{proof}

\begin{proposition}
\label{prop:case}
Let $w_1$ and $w_2$ be two words equal in spaCy. And let $w_1 \neq w_2$. Then for all $j$ such that $w_1(j) \neq w_2(j)$, either $\text{shape}(w_1(j)) = \text{X}$ and $\text{shape}(w_2(j)) = \text{x}$, or $\text{shape}(w_1(j)) = \text{x}$ and $\text{shape}(w_2(j)) = \text{X}$. Furthermore, any such $j$ satisfies: $4 < j < |w_1| - 3 = |w_2| - 3$.
\end{proposition}
\begin{proof}
First note that if either $w_1(j)$ or $w_2(j)$ is a digit or a special character, the norms of the two will be different, contradicting the assumption. This only leaves the possibility that $w_1(j)$ has either shape X or x. In either case $w_2(j)$ should have the case-inverted shape, since $w_1(j) \neq w_2(j)$. From Proposition~\ref{prop:length}, since the prefixes and suffixes of the two words are the same, such a $j$ can only be part of the middle of the two words. This leaves $j \in \{2, 3, 4\}$. Consider one such $j$. Since the length of the prefix is 1, $w_1(j)$ and $w_2(j)$ will not be truncated in the shape. Hence the shapes of the two words will have a different shape at the $j$th character, contradicting the assumption that they are equal in spaCy.
\end{proof}

\begin{proposition}
\label{prop:mid}
Let $w_1$ and $w_2$ be two words equal in spaCy. And let $w_1 \neq w_2$. Then $|\text{mid}(w_1)| = |\text{mid}(w_2)| \geq 5$. 
\end{proposition}
\begin{proof}
From Proposition~\ref{prop:length}, the two words are of equal length. Hence the length of the middle of the two words is also equal. Let us assume that $|\text{mid}(w_1)| < 5$. This means $w_1$ and $w_2$ have at most 4 characters in their middle, and a total length of $\le 8$ including the prefix and suffix. The case that $\text{mid}(w_1)$ is empty is not possible, since then the word only consists of the prefix and suffix, which are the same for the two words. Thus, $0 < |\text{mid}(w_1)| < 5$. From Proposition~\ref{prop:case}, any $j$ satisfying $w_1(j) \neq w_2(j)$ must be from the set $\{5, 6, 7\}$. Consider $j = 5$, and wlog assume the shape of $w_1(5)$ to be X and that of $w_2(5)$ to be x. The shape of the two words up to the $4$th character is the same. Denote this shape by $s$. If the four characters of $s$ are not of the form xxxx, or XXXX, then the $5$th character of $\text{shape}(w_1)$ will be X and that of $\text{shape}(w_2)$ will be x, contradicting the assumption that the two shapes are the same. Thus, wlog assume the four characters of $s$ are $\text{XXXX}$. Then there is a shape change to x at character $5$ in $\text{shape}(w_2)$, but $w_1(5)$ does not change the shape (due to truncation). Since $j = 5$, there are $3$ characters left. This means the shape of $w_1$ will have at least one character less than the shape of $w_2$, since there can be no truncation. This shows that the shape of the two words upto the 5th character is the same. This leaves $j = 6$ and $7$. A similar argument above applied to $j = 6$ and $j = 7$, in turn, leads to the same contradiction. 
\end{proof}
Combining the three propositions, this means that if two words are equal in spaCy then they have the same length. And the only way two words equal in spaCy can be different is if they differ in case, and all the differences are from character 5 onwards. The following is an immediate corollary.
\begin{corollary}
If two words $w_1$ and $w_2$ are the same in spaCy and $|w_1| \le 8$, then $w_1 = w_2$.
\end{corollary}
Thus, the minimum length of words for which the two notions differ is 9. We have already given an example of such a pair.

\section{Memorization differences between spaCy 2 and spaCy 3}

\begin{figure}[t]
    \centering
    \begin{subfigure}{0.9\linewidth}
    \centering
\includegraphics[width=\linewidth]{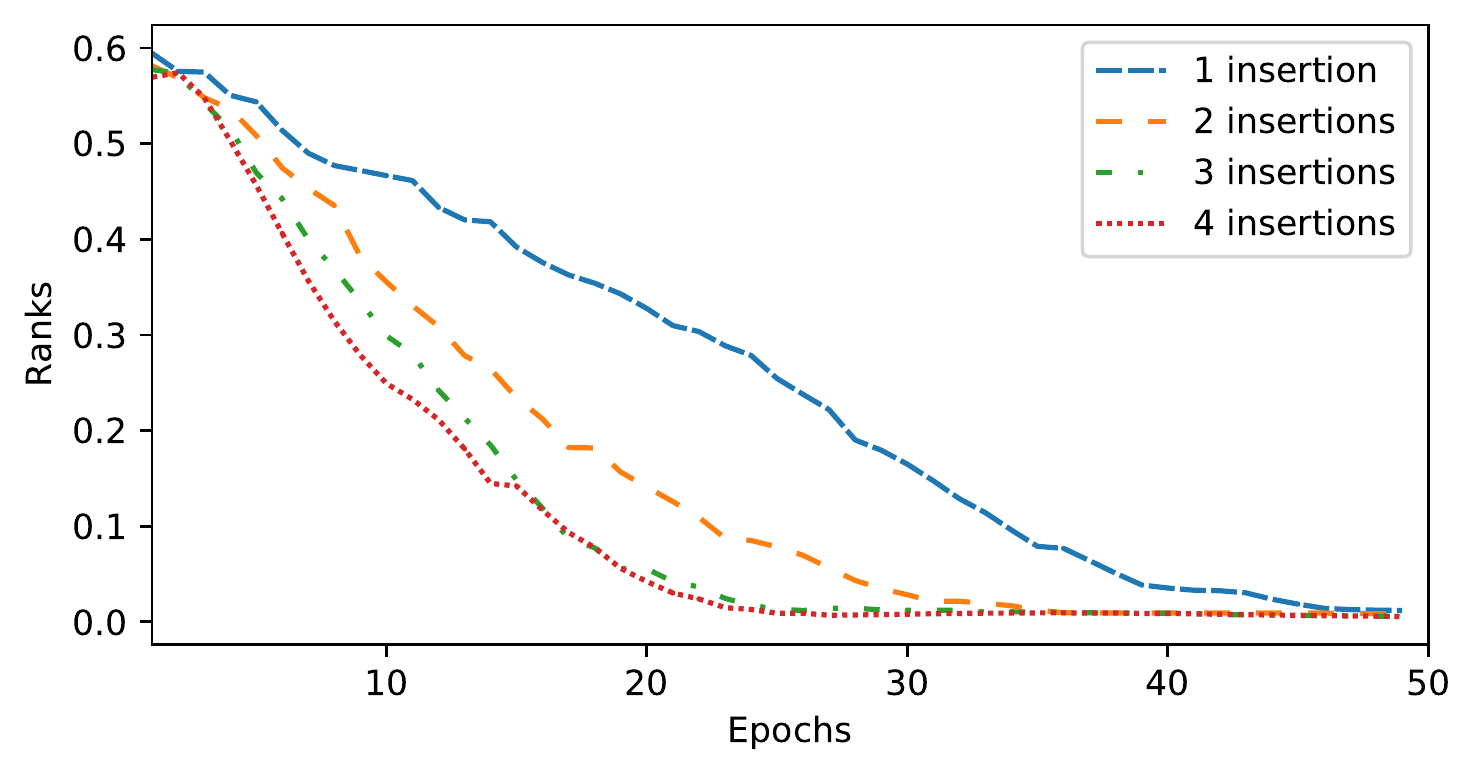}
     \caption{spaCy 2.3.2}\label{fig:app-memorization-spacy2}
    \end{subfigure}
    
    \begin{subfigure}{0.9\linewidth}
    \centering
    \includegraphics[width=\linewidth]{plots/SINGLE_RANK_PER_EPOCH_AND_INSERTION_AVERAGED_LINE_PLOT_spacy3.0.3_XLIM_10.pdf}
    \caption{spaCy 3.0.3}\label{fig:app-memorization-spacy3}
    \end{subfigure}
    \caption{The extent of model memorization for out-of-distribution words in two different versions of spaCy. ``Insertions'' refer to the number of times the same phrase is repeated in the training set, ``Epochs'' is the number of training iterations the model undergoes during the update process.}
    \label{fig:app-memorization-spacy}
\end{figure}

\begin{figure*}[t]
    \centering
    \begin{subfigure}{0.45\textwidth}
    \includegraphics[width=\linewidth]{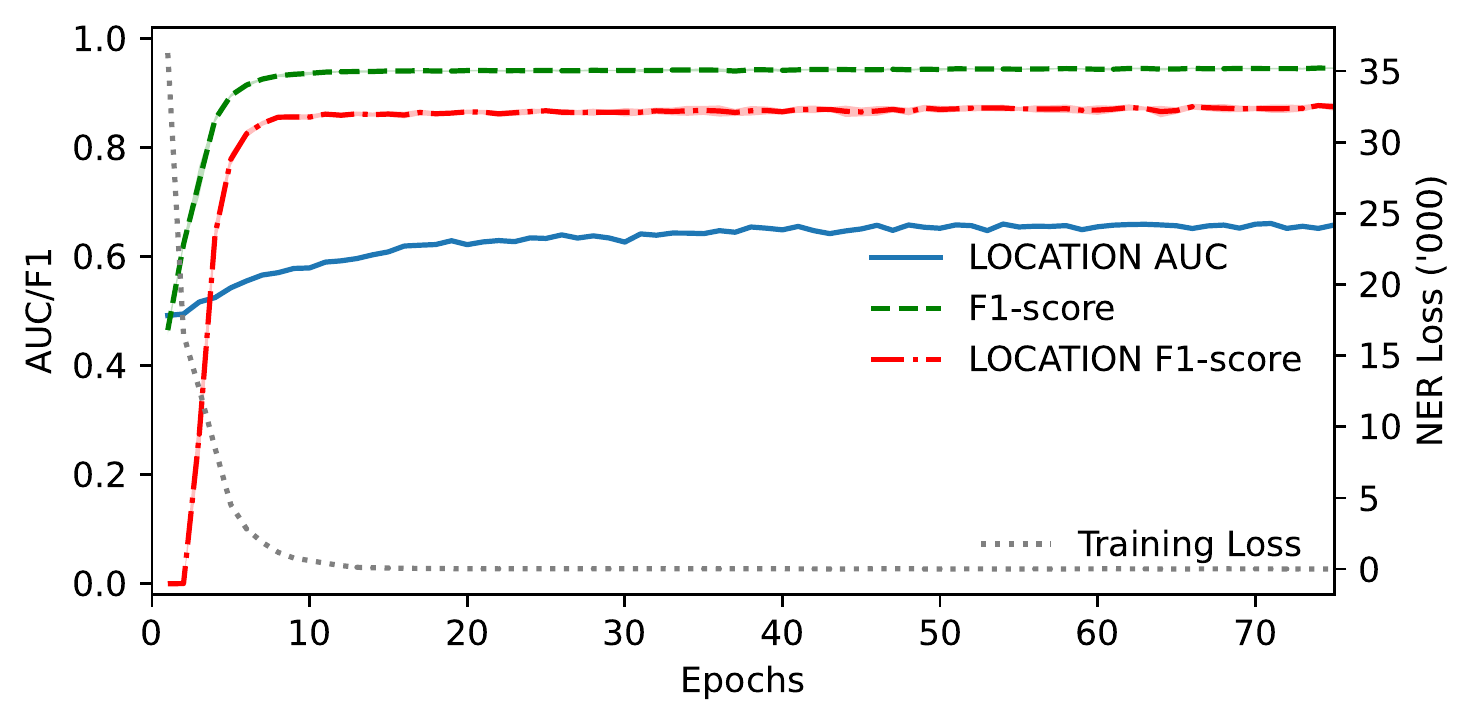}
     \caption{LOCATION}
    \end{subfigure}
    \begin{subfigure}{0.45\textwidth}
    \includegraphics[width=\linewidth]{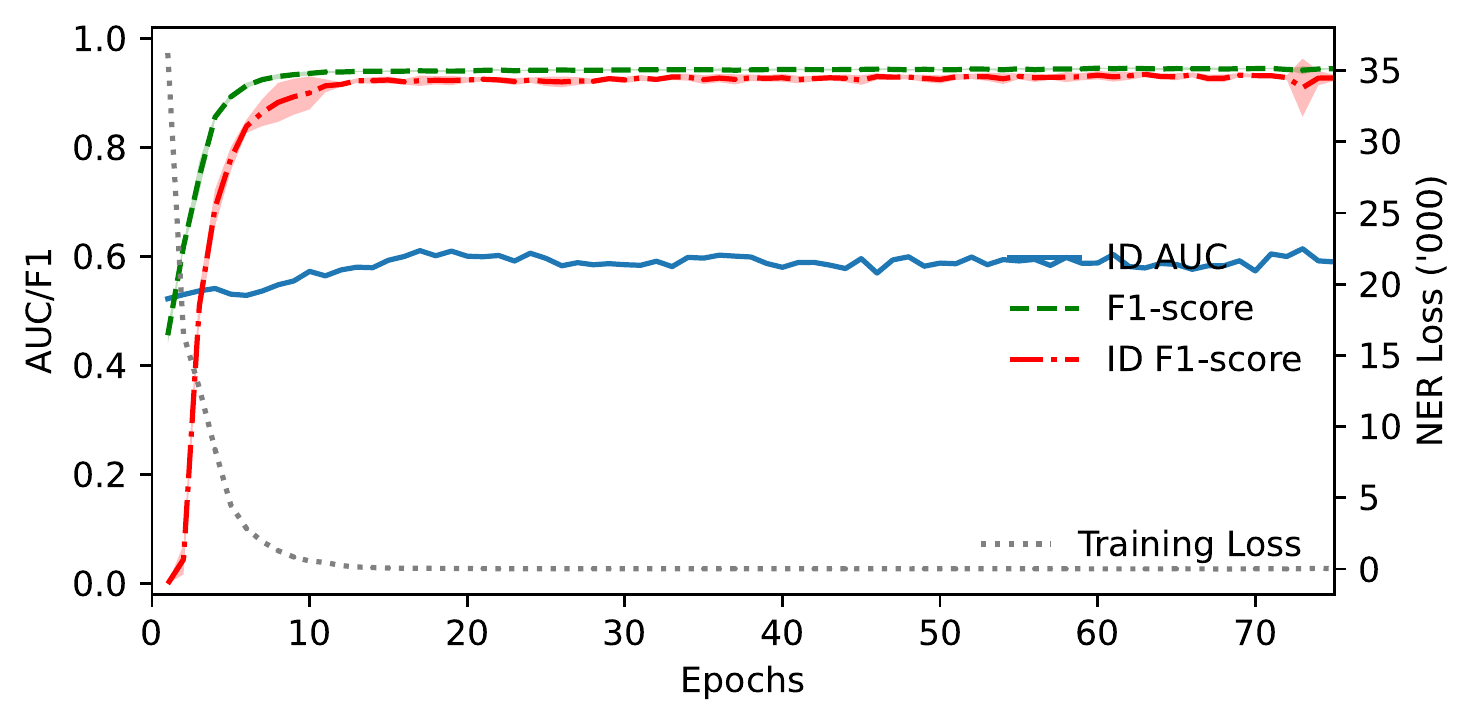}
    \caption{ID}
    \end{subfigure}
    
    \begin{subfigure}{0.45\textwidth}
    \includegraphics[width=\linewidth]{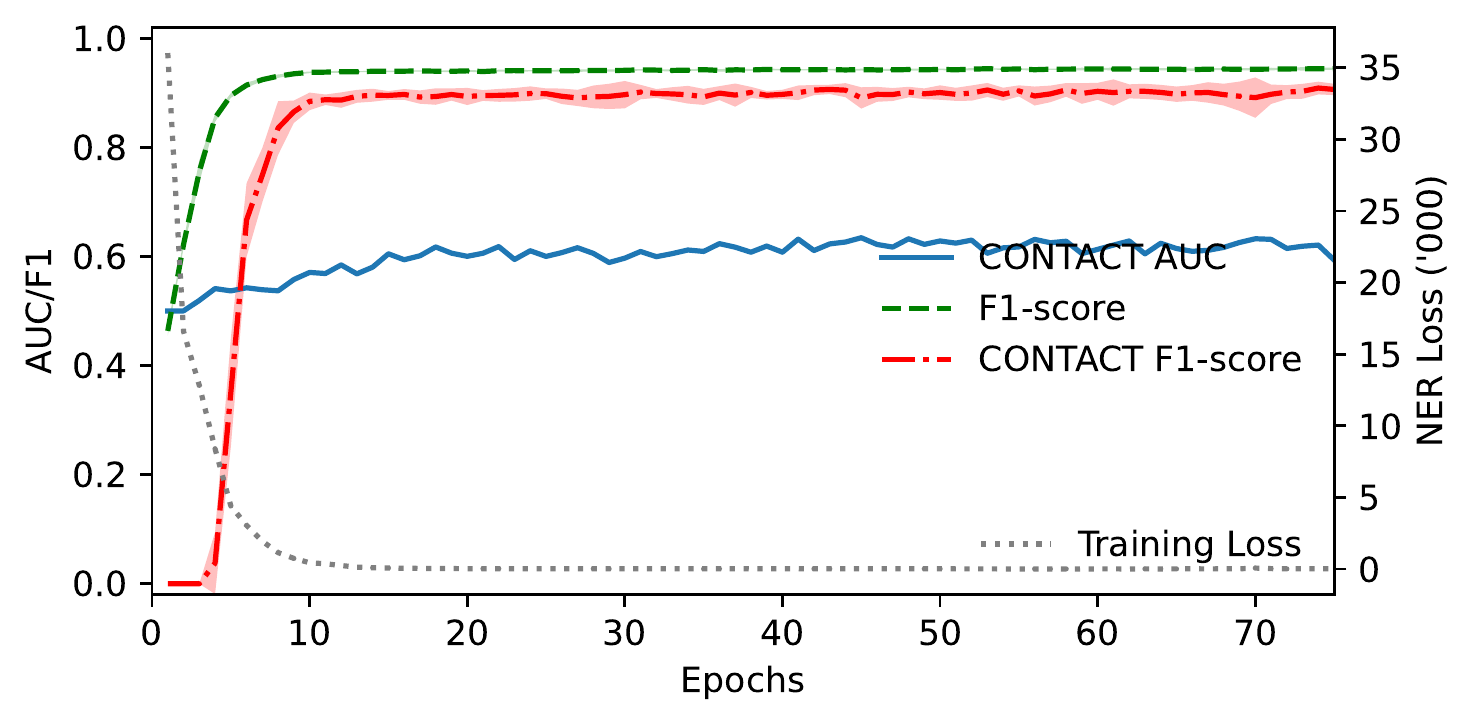}
    \caption{CONTACT}
    \end{subfigure}
    \begin{subfigure}{0.45\textwidth}
    \includegraphics[width=\linewidth]{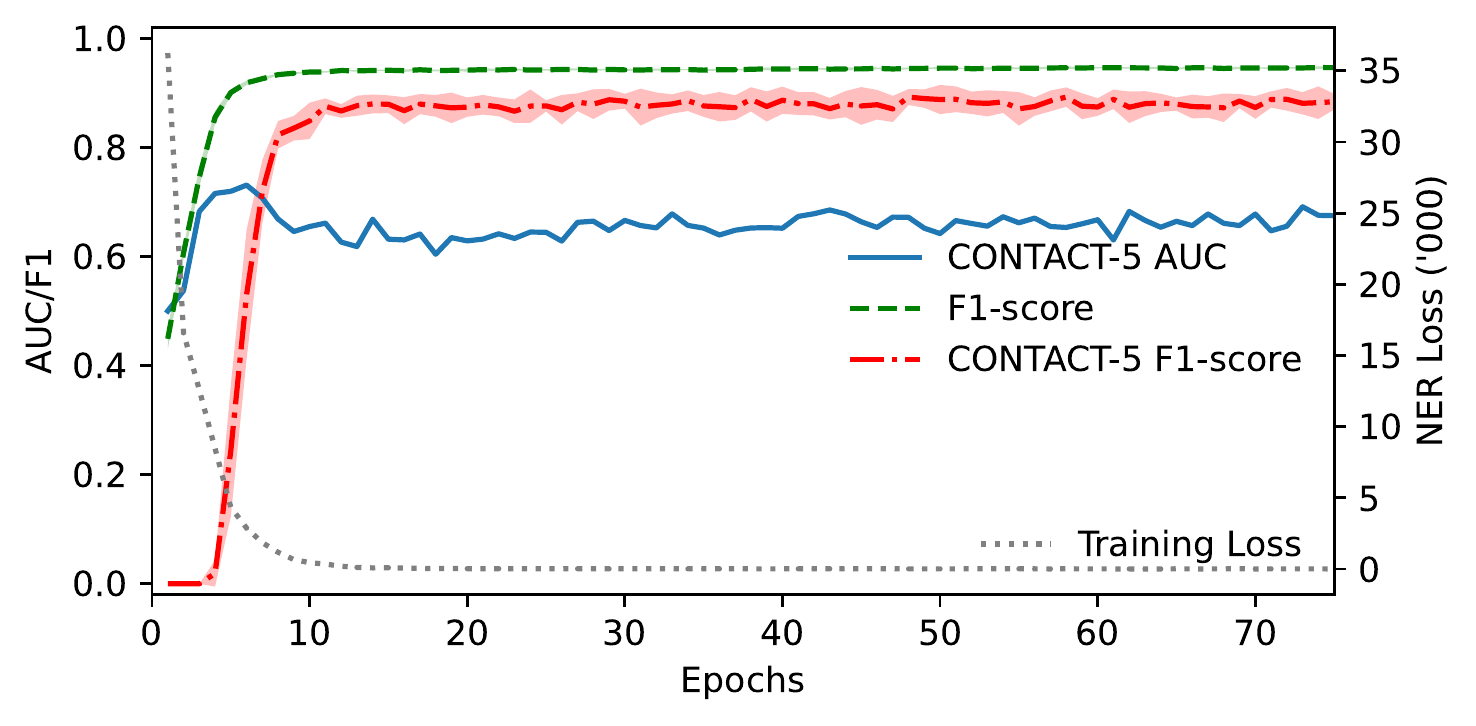}
    \caption{CONTACT-5}
    \end{subfigure}
    
    \begin{subfigure}{0.45\textwidth}
    \includegraphics[width=\linewidth]{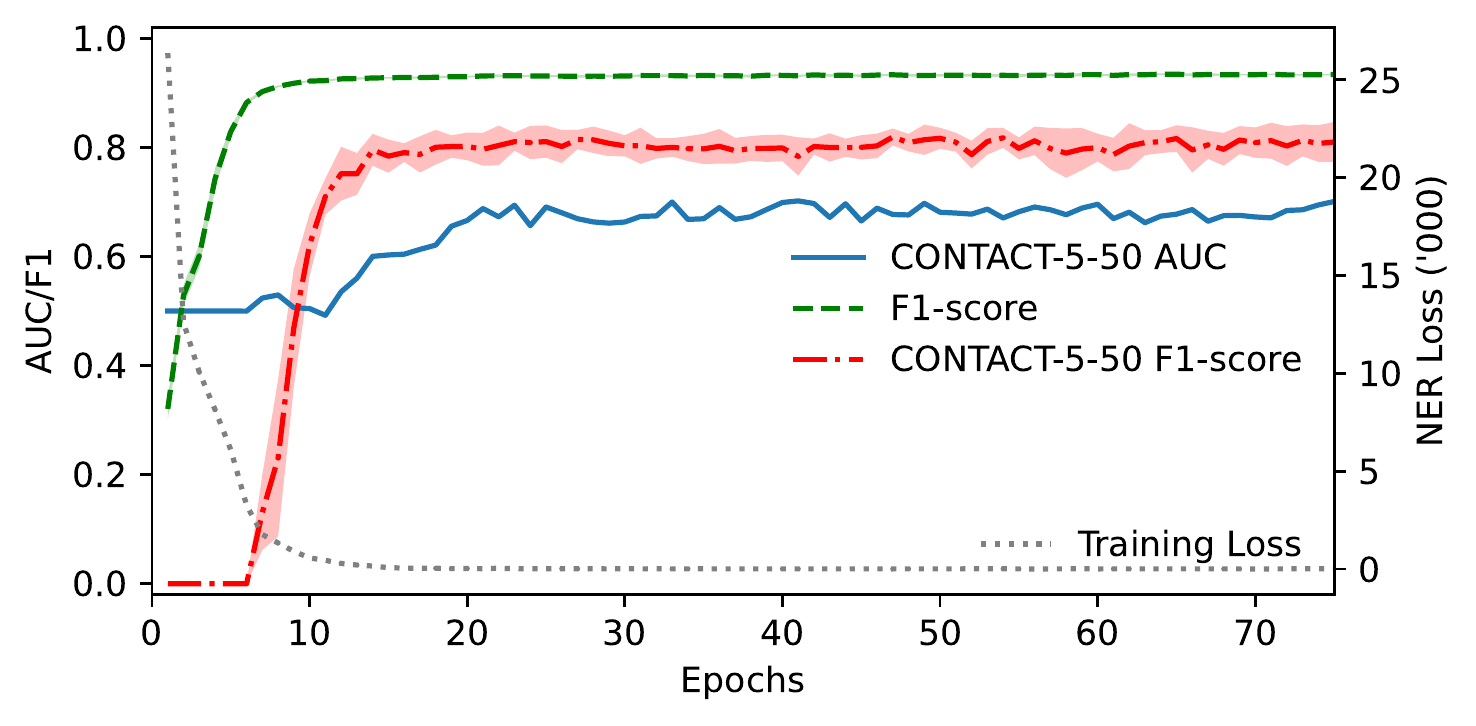}
    \caption{CONTACT-5-50}
    \end{subfigure}
    \begin{subfigure}{0.45\textwidth}
    \includegraphics[width=\linewidth]{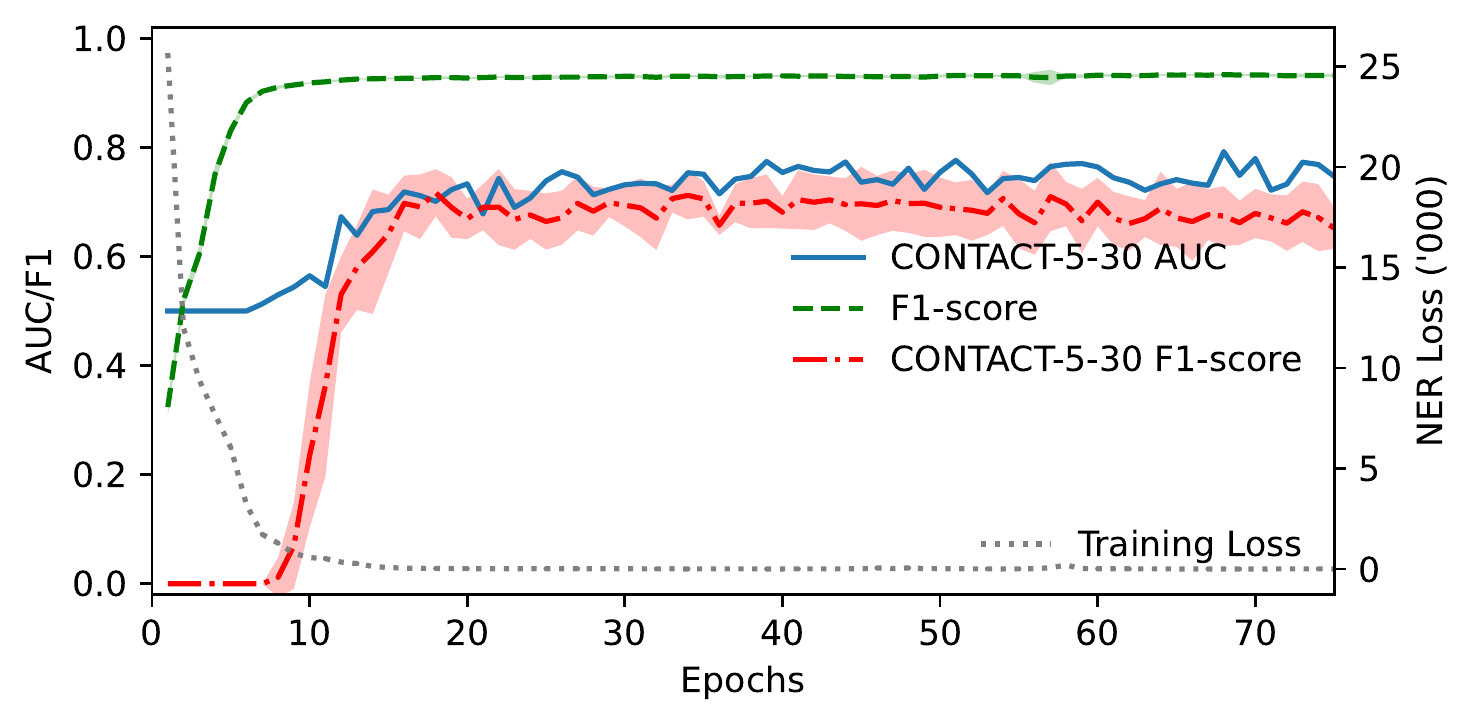}
    \caption{CONTACT-5-30}
    \end{subfigure}
    
    \begin{subfigure}{0.45\textwidth}
    \includegraphics[width=\textwidth]{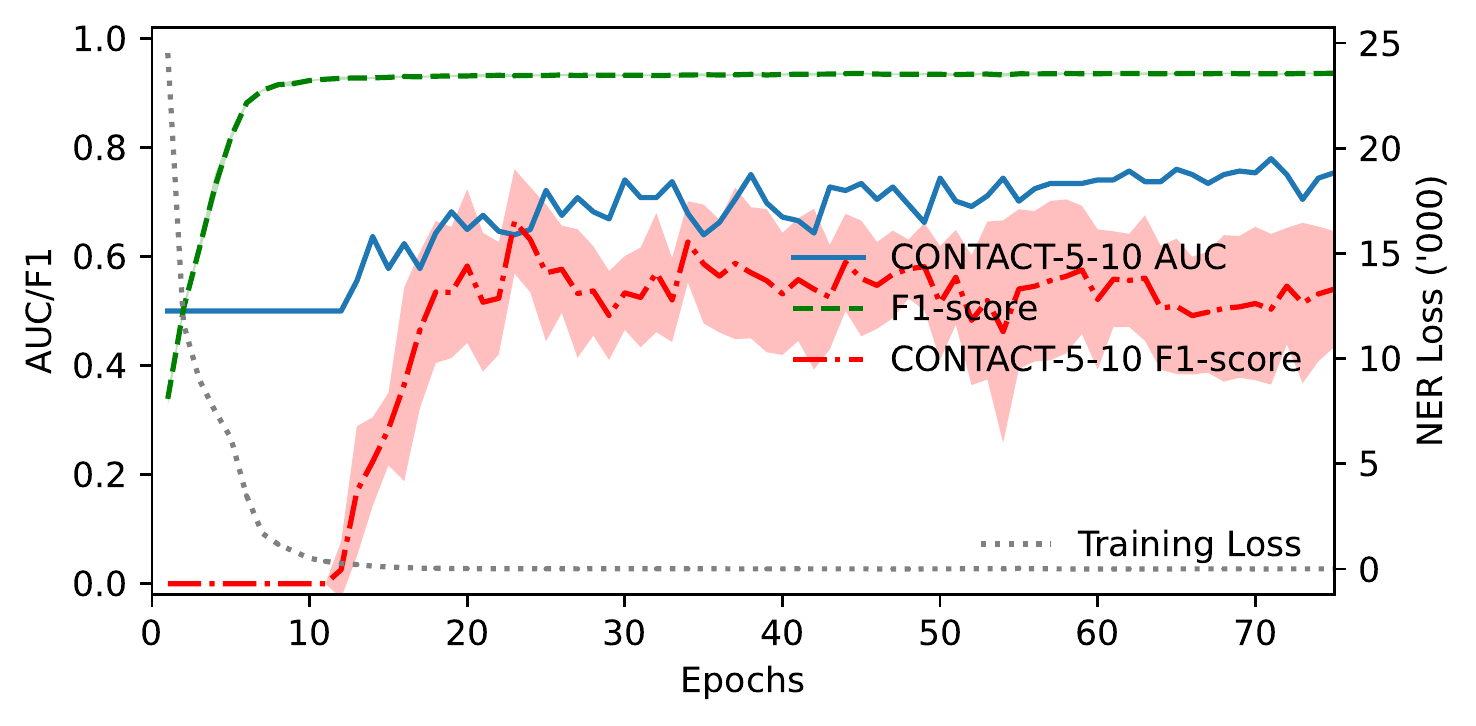}
    \caption{CONTACT-5-10}
    \end{subfigure}
    \caption{Individualized MI AUC, NER Loss and F1-score for each of the entity configurations investigated in Section~\ref{subsec:MIA_medical_dataset}. MI AUC and F1-score share the same y-axis scale on the left.}
    \label{fig:full_i2b2_loss}
\end{figure*}

An additional experiment was performed on two different versions of spaCy: version 3 (which is the current version at the time of writing) and the older, version 2. We use $N = 10$ different target passwords chosen randomly from a space of $|\mathcal{R}| = 2000$ passwords, which in turn are picked randomly from a public password dictionary with one million passwords~\cite{github_passwordlist}.

We then measure the rank of each target password averaged over 48 runs. We vary the number of insertions of the phrase $s[r]$ from 1 to 4. The average (normalized) ranks of the target passwords for spaCy 2 and 3, as a function of the number of epochs used to update the NN model are shown in Figures~\ref{fig:app-memorization-spacy2} and~\ref{fig:app-memorization-spacy3}, respectively. After a certain number of epochs, both versions of spaCy memorize the target password (rank of 0). However, the newer release of spaCy shows faster memorization: less than 10 epochs even with a single insertion. Thus the newer version fairs worse in terms of unintended memorization. Having no substantial architectural differences between models of the two versions, it's interesting to observe a major difference in model memorization on a single phrase. Determining the exact cause for the faster memorization would requite a thorough inspection of the source code.

\section{Complete training loss plots for the document redaction use-case}\label{app:full_docredact_loss}
We present individualized MI AUC, NER training loss and individual entity f1-score for each of the entity configurations investigated in Section~\ref{subsec:MIA_medical_dataset} in Figure~\ref{fig:full_i2b2_loss}.
In Figure~\ref{fig:full_i2b2_secret_loss}, we present individualized MI AUC, NER training loss and individual f1-score for each of the SECRET configurations investigated in Section~\ref{subsubsec:i2b2_secret_redaction}. Similarly, in Figure~\ref{fig:defend_i2b2}, from our defense investigation in replacing training entities with dummy data in Section~\ref{sec:defences}. For each of these figures the standard deviation of the target entity's F1 score is displayed in the shaded red region of the figure.

\begin{figure*}[t]
    \centering
    \begin{subfigure}{0.45\textwidth}
    \includegraphics[width=\linewidth]{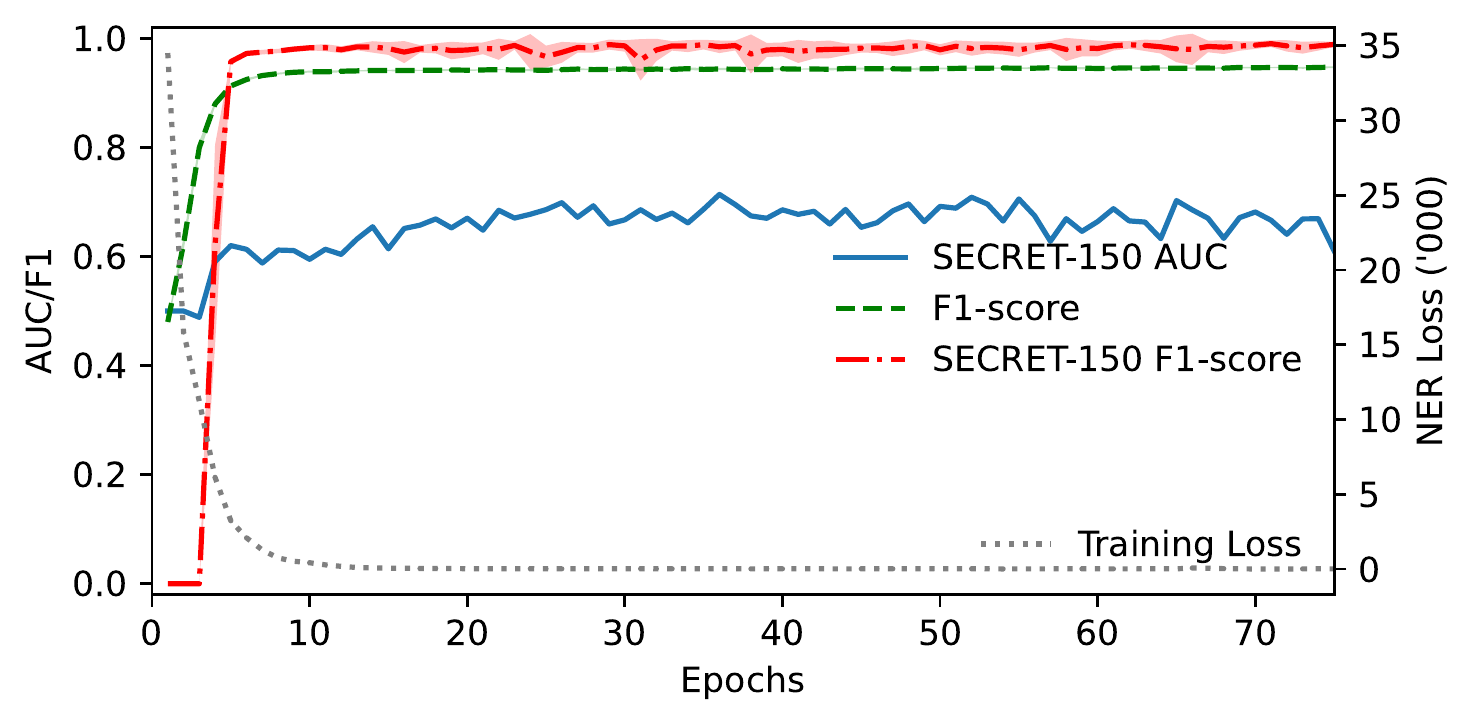}
     \caption{SECRET-150}
    \end{subfigure}
    \begin{subfigure}{0.45\textwidth}
    \includegraphics[width=\linewidth]{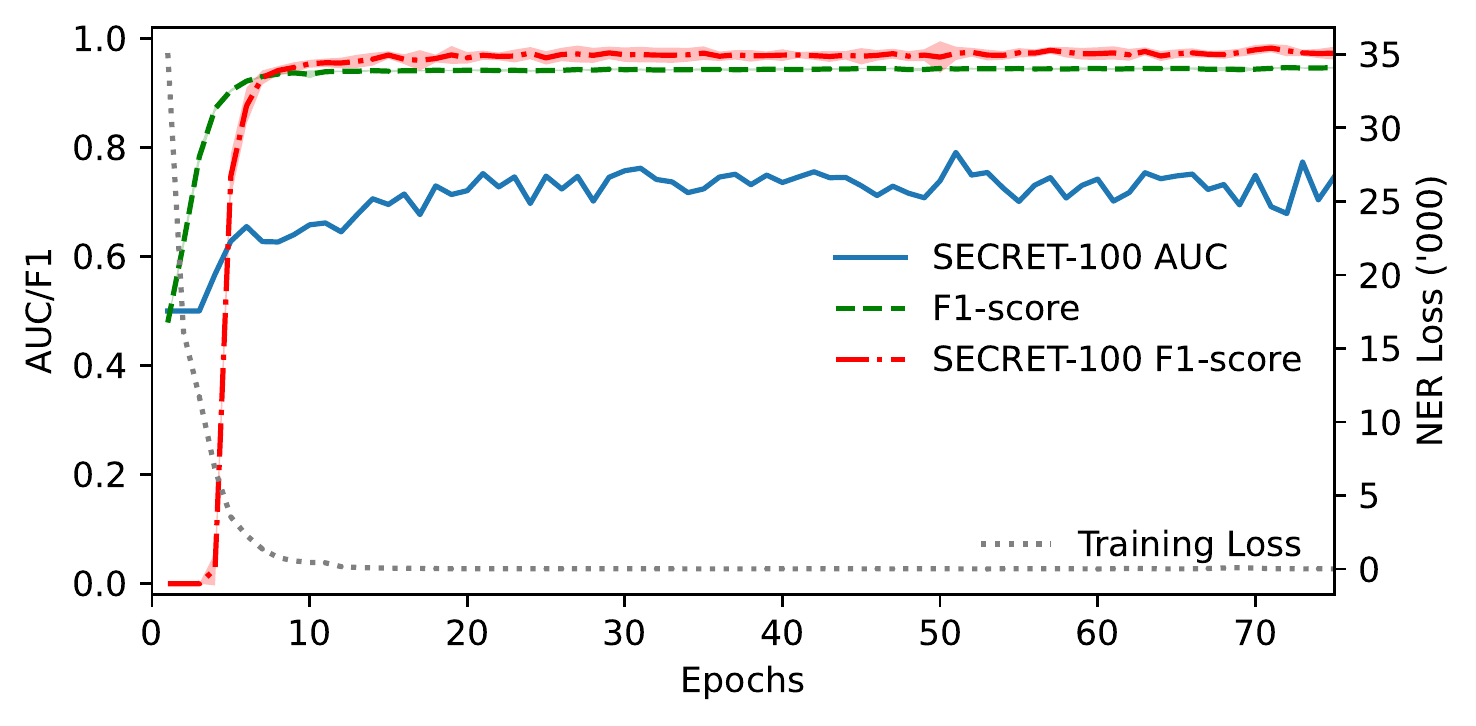}
    \caption{SECRET-100}
    \end{subfigure}
    
    \begin{subfigure}{0.45\textwidth}
    \includegraphics[width=\linewidth]{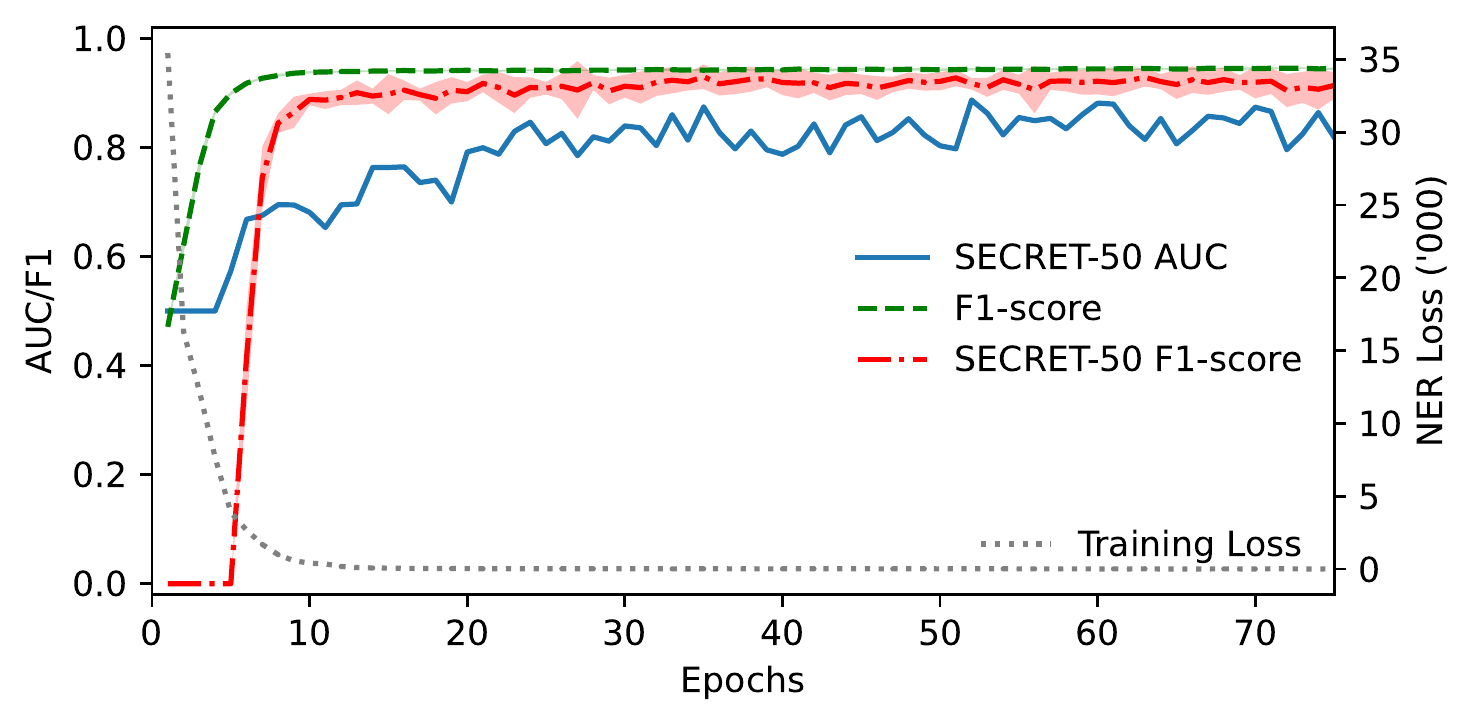}
    \caption{SECRET-50}
    \end{subfigure}
    \begin{subfigure}{0.45\textwidth}
    \includegraphics[width=\linewidth]{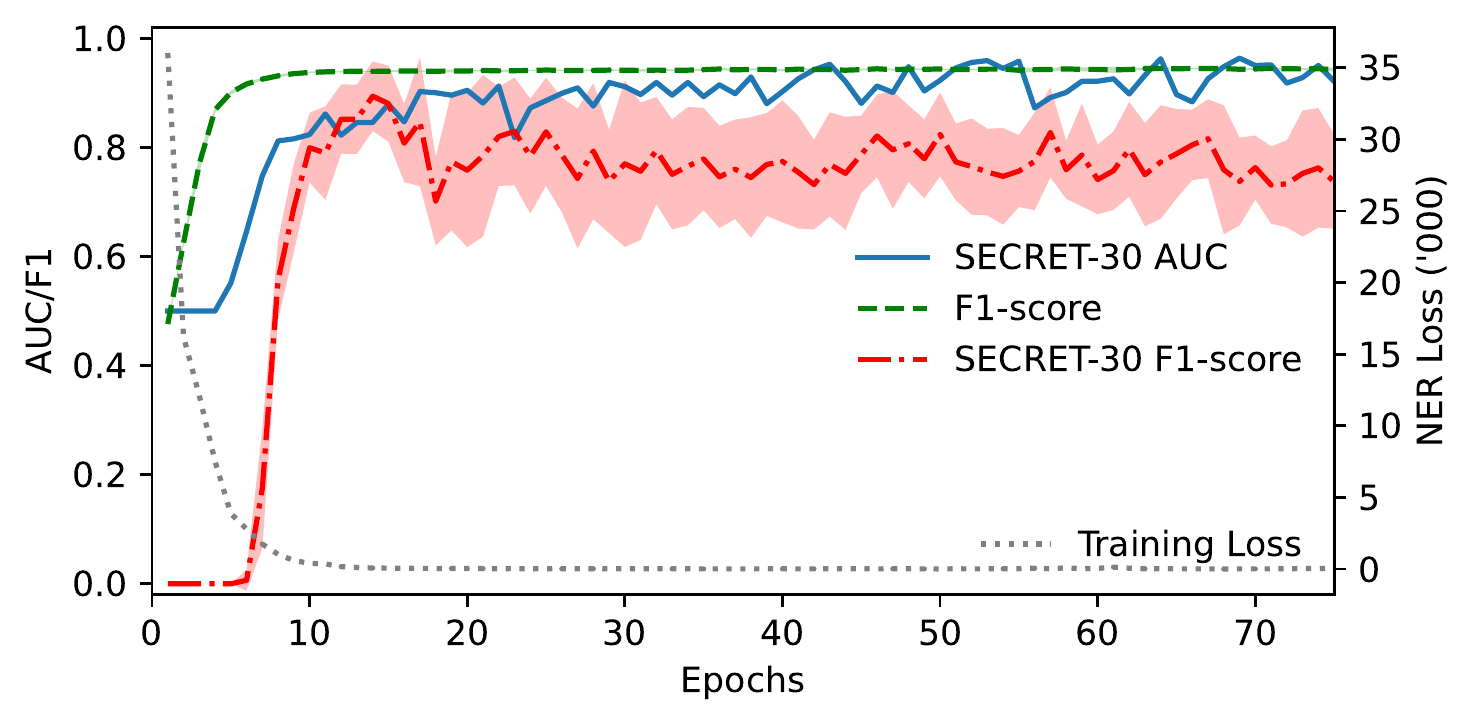}
    \caption{SECRET-30}
    \end{subfigure}

    \caption{Individualized MI AUC, NER Loss and F1-score for the SECRET class with controlled training and test sample configurations investigated in Section~\ref{subsubsec:i2b2_secret_redaction}. MI AUC and F1-score share the same y-axis scale on the left.}
    \label{fig:full_i2b2_secret_loss}
\end{figure*}

\begin{figure*}[t]
    \centering
    \begin{subfigure}{0.45\textwidth}
    \includegraphics[width=\linewidth]{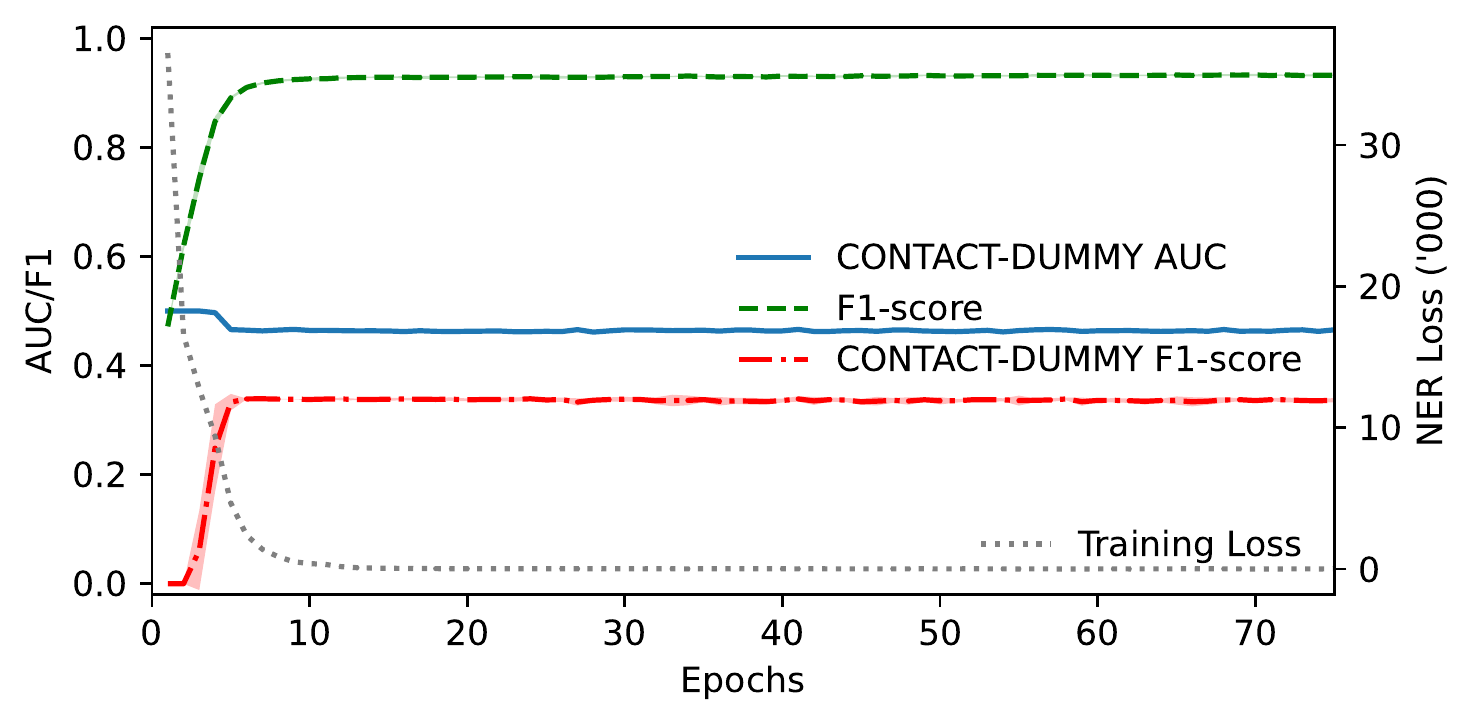}
     \caption{DUMMY CONTACT}
    \end{subfigure}
    \begin{subfigure}{0.45\textwidth}
    \includegraphics[width=\linewidth]{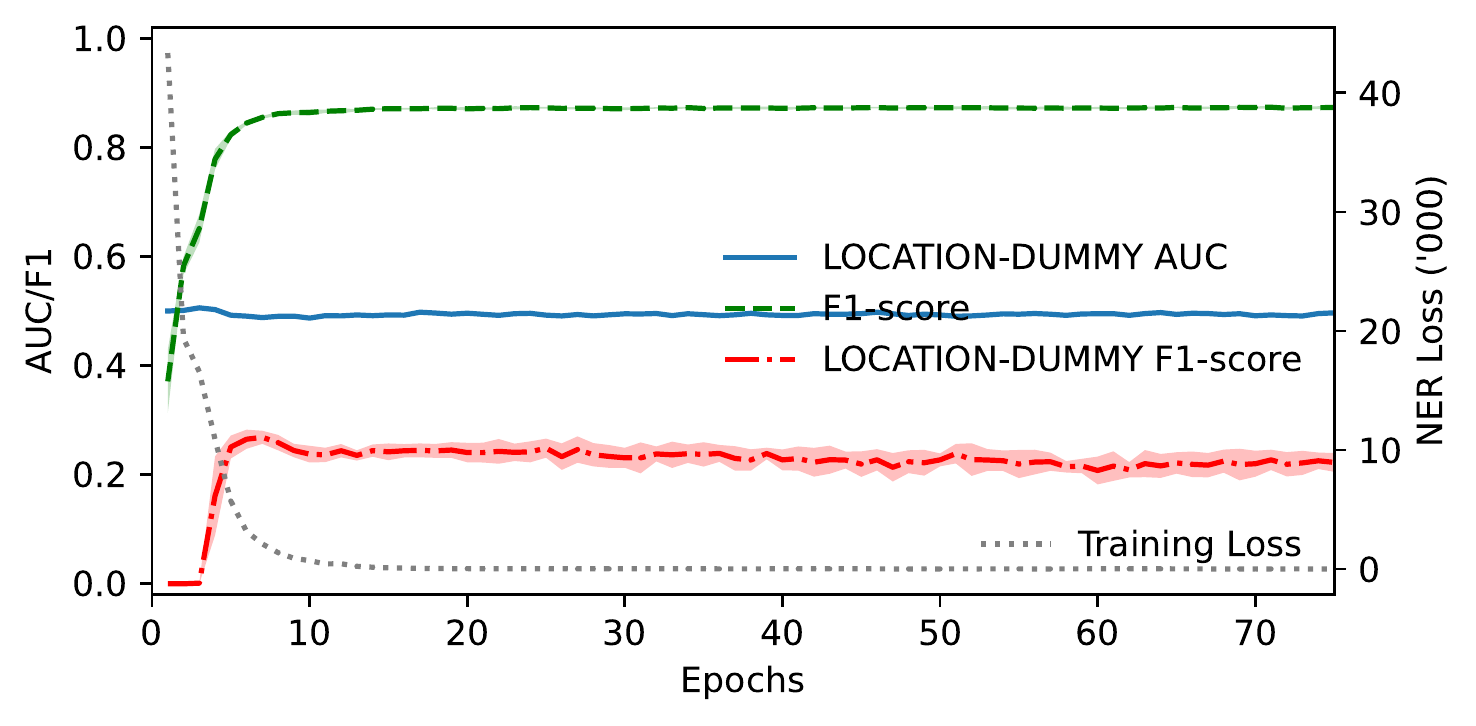}
    \caption{DUMMY LOCATION}
    \end{subfigure}
    \caption{{Individualized MI AUC, NER Loss and F1-score for the defense method investigated in Section~\ref{sec:defences}. MI AUC and F1-score share the same y-axis scale on the left.}}
    \label{fig:defend_i2b2}
\end{figure*}

\section{Timing attack comparison}
\label{sec:app:timing}

\begin{figure*}[t]
\centering
\begin{subfigure}{0.245\textwidth}
\includegraphics[width=\linewidth]{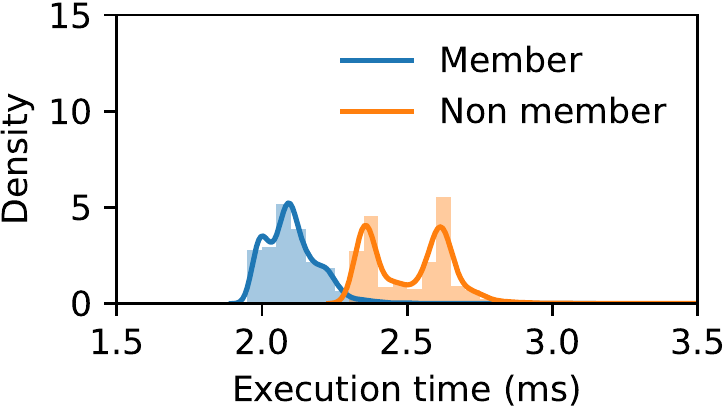}
\caption{LG pipeline, original model}
\label{fig:absolute_time_lg_orig_ner}
\end{subfigure}
\begin{subfigure}{0.245\textwidth}
\includegraphics[width=\linewidth]{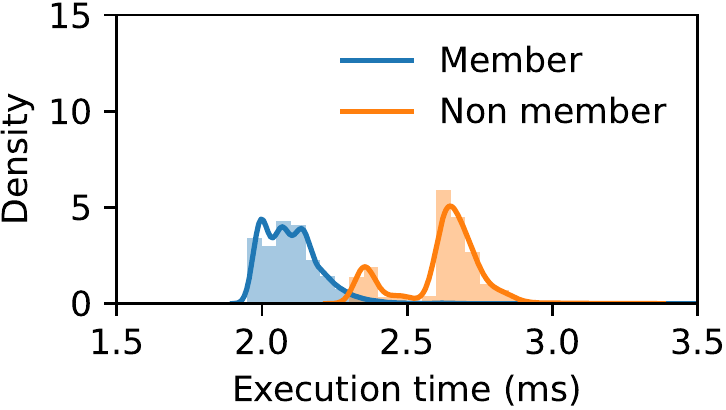}
\caption{TRF pipeline, original model}
\label{fig:absolute_time_trf_orig_ner}
\end{subfigure}
\begin{subfigure}{0.245\textwidth}
\includegraphics[width=\linewidth]{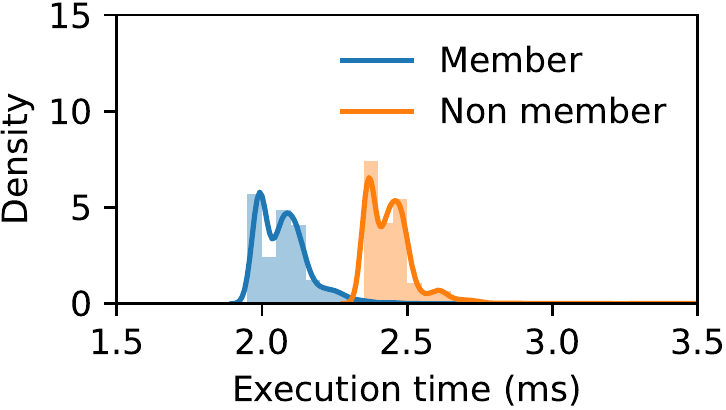}
\caption{LG pipeline, updated model}
\label{fig:absolute_time_lg_update_ner}
\end{subfigure}
\begin{subfigure}{0.245\textwidth}
\includegraphics[width=\linewidth]{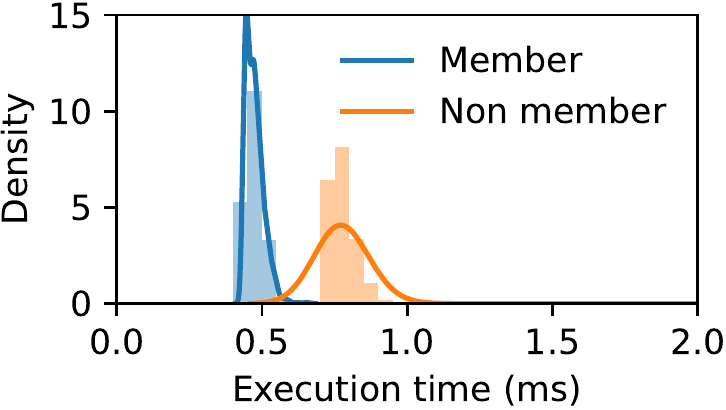}
\caption{TRF pipeline, updated model}
\label{fig:absolute_time_trf_update_ner}
\end{subfigure}
\caption{Execution time distribution of original and updated NER models for processing words inside (member) and outside (non-member) the updating dataset. Results for both a tok2vec (LG) and ROBERTA based transformer (TRF) pipeline is presented.}
\label{fig:absolute_time_ner}
\end{figure*}

In this appendix, we provide a visual representation of the differing timings between member and non-member words when processed by spaCy. Recall that Figure~\ref{fig:absolute_time_lg_orig_ner} \& \ref{fig:absolute_time_trf_orig_ner} shows the execution time distribution of the original pre-trained NER model to process in-vocab and out-vocab words. Next, Figure~\ref{fig:absolute_time_lg_update_ner} \& \ref{fig:absolute_time_trf_update_ner} displays the time distribution of the updated NER model to process passwords inside and outside the updating dataset.

